\documentclass[11pt]{article}
\usepackage[utf8]{inputenc}
\usepackage{amssymb,amsmath,latexsym}
\usepackage[english]{babel}
\usepackage{mathptmx}
\usepackage{graphicx}
\usepackage{subfig}
\usepackage{bibentry}
\usepackage{amsthm}
\usepackage[bookmarks=true,colorlinks=true,linkcolor=blue]{hyperref}
 \usepackage[colorlinks=true]{hyperref}
\hypersetup{urlcolor=blue, citecolor=red}
\usepackage{color,transparent}
 \newtheorem{thm}{Theorem}[section]
 \newtheorem{prop}{Proposition}[section]
 \newtheorem{lem}{Lemma}[section]
 \newtheorem{exm}{Exemple}[section]

 \newtheorem{dfn}{Definition}[section]

 \newtheorem{rem}{Remark}[section]
 \numberwithin{equation}{section}

\newcommand{\biindice}[3]%
{

\begin{array}[t]{c}
#1\\
{\scriptstyle #2}\\
{\scriptstyle #3}
\end{array}

}
\date{}
\begin{document}
\title{Bohr-Sommerfeld Quantization Rules for 1-D Semiclassical Pseudo-Differential Operator: the Method of Microlocal Wronskian and Gram Matrix}
\author{Abdelwaheb Ifa}
\maketitle
\centerline{University of Tunis El Manar, Faculty of Sciences of Tunis, Research Laboratory of Partial}
\centerline{Differential Equations (LR03ES04), Campus Universitaire El-Manar, 2092 El Manar Tunis}
\centerline{\& University of Kairouan, Higher Institute of Applied Sciences and}
\centerline{Technology of Kairouan, 3100 Kairouan, Tunisia}
\centerline{email: \url{abdelwaheb.ifa@fsm.rnu.tn}}
%\tableofcontents
%\cleardoublepage%
%\tableofcontents
%\cleardoublepage%
%\date{}%
\begin{abstract}
In this paper, we revisit the well known Bohr-Sommerfeld quantization rule (BS) of order 2 for a self-adjoint 1-D semiclassical pseudo-differential operator, within the algebraic and microlocal framework of B. Helffer and J. Sj\"{o}strand. BS holds precisely when the Gram matrix consisting of scalar products of certain WKB solutions with respect to the "flux norm" is not invertible. This condition is obtained using the microlocal Wronskian and does not rely on traditional matching techniques. It is simplified by using action-angle variables. The interest of this procedure lies in its possible generalization to matrixvalued Hamiltonians, like BdG Hamiltonian.
\end{abstract}
\section{Introduction}\label{S1}
Bohr-Sommerfeld quantization rule, in its first formulation, allows to compute the energy levels $E$ of a particle in a one-dimensional potential well, its dynamics being described by the semi-classical Schr\"{o}dinger operator
$$P(x,hD_{x})=(hD_{x})^2+V(x),\quad \text{with } D_x = -i \frac{d}{dx}.$$
It is given at first order in $h$ by the well known formula
$$\frac{1}{2\pi h}\,\oint_{\gamma_E}\xi(x)\,dx=n+\frac{1}{2}$$
Here $\xi(x)=\sqrt{E-V(x)}$ denotes the momentum of the particle on its orbit $\gamma_E\subset T^*\mathbb{R}$ above the potential well $\{V(x)\leq E\}$, $n$ is an integer and the integral is computed over $\gamma_E$ in the phase space $T^{*}\mathbb{R}$.

By the implicit function theorem, we then find $E=E_n(h)$. In other words, the number of wavelengths (associated with the particle along $\gamma_E$ by de Broglie correspondence) must be an integer plus 1/2, called \emph{Maslov correction}.

Let \(p(x,\xi;h)\) be a smooth real classical Hamiltonian on \(T^*\mathbb{R}\), admitting a semiclassical expansion
\[p(x,\xi;h)\sim p_0(x,\xi)+h\,p_1(x,\xi)+h^{2}\,p_2(x,\xi)+\cdots,\quad h\to0,\]
where \(p_0\) is the principal symbol, and \(p_1\) the sub-principal symbol. We assume that \(p\in S^0(m)\) for some order function \(m\), and that \(p+i\) is elliptic. Here, \(S^0(m)\) denotes the class of symbols satisfying
\[|\partial_x^\alpha \partial_\xi^\beta p(x,\xi;h)|\leq C_{\alpha,\beta} m(x,\xi),\quad\forall\,\alpha,\beta\in\mathbb{N},\,(x,\xi)\in T^*\mathbb{R},\]
for some order function \(m(x,\xi)\geq 1\), uniformly in \(h\in(0, h_0]\). This allows to take the Weyl quantization of \(p\), namely
\begin{equation}\label{A01WEyLOO}
P(x,hD_x;h)u(x; h)=\big(Op^{W}(p)(u)\big)(x)=(2\,\pi\,h)^{-1}\,\int\int_{\mathbb{R}\times \mathbb{R}}
e^{\frac{i}{h}\,(x-y)\,\eta}\,p\big(\frac{x+y}{2},\eta\big)\,u(y)\,dy\,d\eta.
\end{equation}
We make the geometric assumption (H) of~\cite{CdV1}: fix a compact interval \(I=[E_-, E_+]\), and assume that there exists a topological ring \(\mathcal{A}\subset p_0^{-1}(I)\), such that \(\partial \mathcal{A}=A_-\cup A_+\), with \(A_{\pm}\) connected components of \(p_0^{-1}(E_{\pm})\), and that \(p_0\) has no critical point in \(\mathcal{A}\). Moreover, \(A_-\) is included in the disk bounded by \(A_+\). (if it is not the
case, we can always change $p$ to \(-p\)). These conditions guarantee that the spectrum of \(P\) in \(I\) is discrete and can be described semiclassically.

We define the microlocal well \(W\) as the disk bounded by \(A_+\). For each \(E\in I\), let \(\gamma_E \subset W\) be a periodic orbit on the energy surface \(\{ p_0(x,\xi)=E\}\), so that \(\gamma_E\) is an embedded Lagrangian manifold.

Let $\mathcal{K}_{h}^{N}(E)$ be the microlocal kernel of $P-E$ of order $N$, i.e. the space of local solutions of $(P-E)u=\mathcal{O}(h^{N+1})$ in the distributional sense, microlocalized on $\gamma_{E}$. This is a smooth complex vector bundle over $\pi_{x}(\gamma_{E})$. Here we address the problem of finding the set of $E=E(h)$ such that $\mathcal{K}_{h}^{N}(E)$ contains a global section, i.e. of constructing a sequence of quasi-modes $\big(u_{n}(h), E_{n}(h)\big)$ of a given order $N$. As usual we denote by $\mathcal{K}_{h}(E)$ the microlocal kernel of $P-E$ mod $\mathcal{O}(h^{\infty})$; since the distinction between $\mathcal{K}_{h}^{N}(E)$ and $\mathcal{K}_{h}(E)$ plays no important role here, we shall content to write $\mathcal{K}_{h}(E)$.

Then if $E_{+}<E_{0}=\displaystyle\liminf_{|x,\xi|\rightarrow +\infty}p_{0}(x,\xi)$, all eigenvalues of $P$ in $I$ are indeed by \emph{Bohr-Sommerfeld quantization condition} (BS) $\mathcal{S}_{h}\big(E_{n}(h)\big)=2\pi n h$, where the semiclassical action $\mathcal{S}_{h}(E)$ has the asymptotics
$$\mathcal{S}_{h}(E)\sim S_{0}(E)+h\,S_{1}(E)+h^{2}\,S_{2}(E)+\cdots$$
We determine BS at any accuracy by computing quasi-modes. There are a lot of ways to derive BS: the method of matching of WKB solutions \cite{Bender}, known also as Liouville-Green method \cite{Olver}, which has received many improvements \big(see \cite{Yafaev}\big), the method of the monodromy operator \big(see \cite{HR84} and references therein\big), the method of quantization deformation based on Functional Calculus and Trace Formulas \cite{Littlejohn0}, \cite{CdV1}, \cite{CARGooGR}, \cite{Gracia00Saz}, \cite{Argyres}.

Note that the latter one already assumes BS, it only gives a very convenient way to derive it. In the real analytic case, BS rule, and also tunneling expansions, can be obtained using the so-called "exact WKB method" see e.g. \cite{Fedoriouk}, \cite{DePh}, \cite{DeDiPh} when $P(x,hD_x)=(hD_x)^{2}+V(x)$ is Schr\"{o}dinger operator.
\paragraph{}
Here we present another way to construct quasi-modes of order 2, based on \cite{JO22025}, \cite{HARPE0RIII}. We stress that our method in the present scalar case, when carried to second order, is a bit more intricated than \cite{Littlejohn0}, \cite{CdV1} and its refinements \cite{Gracia00Saz}; it is most useful for matrix valued operators with double characteristics such as Bogoliubov-de Gennes Hamiltonian \big(\cite{BenIfaRo}, \cite{BeMhRo}, \cite{Duncan}\big), or Born-Oppenheimer type Hamiltonians \big(\cite{Baklouti}, \cite{Rouleux}\big).
\begin{exm}(BS quantization rule for the Harmonic Oscillator on $\mathbb{R}$)\ \\
We consider the one-dimensional quantum harmonic oscillator given by the semiclassical operator:
$$P_{0}(x,hD_{x})=\frac{1}{2}\,\left(x^2+(h D_x)^2\right)$$
where $h$ is the semiclassical parameter. The associated classical Hamiltonian (principal symbol) is:
$$p_0(x, \xi)=\frac{1}{2}\,(x^2+\xi^2)$$
In the phase space $(x, \xi)$, the energy surface $\{p_0(x, \xi) = E\}$ is a circle of radius $\sqrt{2E}$. The classical motion along this orbit is described by:
\[x(t)=\sqrt{2E}\,\cos(t),\quad \xi(t)=-\sqrt{2E}\,\sin(t),\]
which is a closed trajectory with period $T=2\pi$.
To apply the Bohr-Sommerfeld quantization rule, we compute the classical action $S_0(E)$, which is the integral of the momentum along the closed orbit $\gamma_E$ at energy $E$:
\[S_0(E)=\oint_{\gamma_E}\xi\,dx.\]
Using the parametrization of the orbit, we compute:
\begin{align*}
S_0(E)
&=\int_0^{2\pi}\xi(t)\,\dot{x}(t)\,dt \\
&=\int_0^{2\pi}\left(-\sqrt{2E}\sin t\right)\cdot\left(-\sqrt{2E}\sin t\right)\,dt \\
&=2E \int_0^{2\pi}\sin^2 t\,dt\\
&=2E\cdot\pi=2\pi E.
\end{align*}
The Bohr-Sommerfeld quantization condition, including the Maslov correction, reads:
\[S_0(E)+h\,S_1(E)=2\pi h n,\quad n\in\mathbb{N}.\]
For the harmonic oscillator, the Maslov index is $\mu = 2$, which gives:
\[S_1(E)=-\frac{\pi}{2}.\]
This leads to the equation:
\[2\pi E-\pi h=2\pi h n \quad \Rightarrow \quad E=h\left(n + \frac{1}{2}\right).\]
Therefore, the quantized energy levels of the harmonic oscillator are:
\[E_n = h\left(n + \frac{1}{2}\right), \quad n = 0, 1, 2, \ldots\]
These are exactly the energy levels of the quantum harmonic oscillator.
\end{exm}
\section{The microlocal Wronskian}\label{S2}
The best algebraic and microlocal framework for computing 1-D quantization rules in the self-adjoint case, developed in the fundamental works of \cite{JO22025}, \cite{HARPE0RIII}, is based on Fredholm theory and the classical \emph{positive commutator method}, which involves conservation of a quantity called \emph{quantum flux}.
\paragraph{}
Bohr-Sommerfeld quantization rules are derived by constructing quasi-modes using the WKB approximation along a closed Lagrangian manifold \(\Lambda_E\subset \{p_0=E\}\), i.e. a periodic orbit of the Hamiltonian vector field \(H_p\) with energy \(E\). This construction is local and depends on the rank of the projection \(\Lambda_E \rightarrow\mathbb{R}_x\).
\paragraph{}
Thus, the set \(K_h(E)\) of microlocal solutions to \((P-E)u=0\) along \(\Lambda_E\) can be seen as a bundle over $\mathbb{R}$ with a compact base, corresponding to the classically allowed region at energy \(E\). The eigenvalues \(E_n(h)\) are then determined by the condition that the global quasi-mode obtained by gluing local WKB solutions along \(\Lambda_E\) is singlevalued, i.e. that \(K_h(E)\) has trivial holonomy.
\paragraph{}
Assuming \(\Lambda_E\) is smoothly embedded in \(T^*\mathbb{R}\), it can always be parametrized by a non-degenerate phase function. Of particular interest are the \emph{focal points}, i.e. critical points of the phase functions, which are responsible for the change in Maslov index. A point \(a_{E}=(x_E, \xi_E)\in\Lambda_E\) is a focal point if \(\Lambda_E\) "turns vertical" at \(a_{E}\), meaning that the tangent space \(T_{a_{E}}\Lambda_E\) is no longer transverse to the fiber \(x=\text{const}\). in \(T^*\mathbb{R}\).
\paragraph{}
In any case however, \(\Lambda_E\) can locally be parametrized either by a phase function \(S(x)\) (spatial representation) or \(\widetilde{S}(\xi)\) (Fourier representation). We fix an orientation on \(\Lambda_E\) and for any point \(a\in\Lambda_E\) (not necessarily a focal point), we denote by \(\rho=\pm1\) the oriented segments near \(a\). Let \(\chi^{a}\in C_0^\infty(\mathbb{R}^2)\) be a smooth cut-off function equal to 1 near \(a\), and \(\omega^{a}_{\rho}\) a small neighborhood of \(\operatorname{supp}[P,\chi^{a}]\cap\Lambda_E\) near \(\rho\). Here, \(\chi^{a}\) holds for \(\chi^{a}(x,hD_x)\) as in (\ref{A01WEyLOO}), and we shall equally write $P(x,hD_{x})$ in spatial representation, or $P(-h D_\xi,\xi)$ in Fourier representation.
\begin{dfn}\label{A02}
Let \(P\) be self-adjoint and \(u_a,v_a\in K_h(E)\) be microlocal solutions supported on \(\Lambda_E\). We define the \emph{microlocal Wronskian} of \((u^{a}, \overline{v^{a}})\) near \(a\) in $\omega^{a}_{\rho}$ as
\begin{equation}\label{A03}
\mathcal{W}^{a}_{\rho}(u^{a},\overline{v^{a}})=\big(\frac{i}{h}\,[P,\chi^{a}]_{\rho}u^{a}|v^{a}\big)
\end{equation}
where \(\frac{i}{h}[P,\chi^{a}]_{\rho}\) denotes the part of the commutator supported microlocally on \(\omega^{a}_{\rho}\).
\end{dfn}
To clarify the meaning of this definition, consider the Schrödinger operator \(P(x,hD_{x})=(hD_{x})^{2}+V(x)\), with \(x_E=0\), and take \(\chi\) to be the Heaviside step function \(\chi(x)\). Then, in the distributional sense,
\[\frac{i}{h}[P,\chi]=-i\,h\,\chi''+2\,\chi'\,hD_x=-i\,h\,\delta'+2\,\delta\,hD_x,\]
so that $\big(\displaystyle {i\over h}[P,\chi]u|u\big)=-ih\big(u'(0)\overline{u(0)}-u(0)\overline{u'(0)}\big)$
which is the usual Wronskian of $(u, \overline{u})$.
\begin{prop}
Let \(u^{a}, v^{a}\in K_{h}(E)\), and denote by \(\widehat{u}\) the $h$-Fourier (unitary) transform of \(u\). Then:
\begin{equation}\label{A04}
\mathcal{W}=\big(\frac{i}{h}\,[P,\chi^{a}] u^{a}|v^{a}\big)=\big(\frac{i}{h}\,[P,\chi^{a}] \widehat u^{a}|\widehat v^{a}\big)=0
\end{equation}
and
\begin{equation}\label{A05}
\mathcal{W}^{a}_{+}(u^{a},\overline{v^{a}})=-\mathcal{W}^{a}_{-}(u^{a},\overline{v^{a}})
\end{equation}
\big(all equalities being understood mod $\mathcal{O}(h^{\infty})$, resp $\mathcal{O}(h^{N+1})$ when considering $u^{a}, v^{a}\in K_{h}(E)\big)$. Moreover, \(\mathcal{W}^{a}_{\rho}(u^{a}, \overline{v^{a}})\) does not depend modulo \(\mathcal{O}(h^\infty)\) \big(resp. $\mathcal{O}(h^{N+1})$\big) on the choice of \(\chi^{a}\) above.
\end{prop}
\begin{proof}
Since $u^{a}, v^{a}\in K_{h}(E)$ are distributions in $L^{2}$, the first equality (\ref{A04}) follows from the Plancherel formula and the regularity of microlocal solutions in $L^{2}$, $p+i$ being elliptic. If $a$ is not a focal point,  $u^{a}, v^{a}$ are smooth WKB solutions near $a$, so we can expand the commutator in $\mathcal{W}=\big(\frac{i}{h}\,[P,\chi^{a}] u^{a}|v^{a}\big)$ and use that $P$ is self-adjoint to show that $\mathcal{W}=\mathcal{O}(h^{\infty})$. If $a$ is a focal point, $u^{a}, v^{a}$ are smooth WKB solutions in Fourier representation, so again $\mathcal{W}=\mathcal{O}(h^{\infty})$. Then (\ref{A05}) follows from Definition \ref{A02}.
\end{proof}
\section{Second-order BS quantization for a self-adjoint 1-D $h$-PDO}\label{S3}
We apply the method of the microlocal Wronskian and Gram matrix,%described in Chapter~1 (Section~1.2),%
to derive BS quantization conditions at order 2 for a $h$-PDO of the type (\ref{A01WEyLOO}). To simplify, we assume that the principal symbol $p_0$ contains only two focal points along the classical orbit $\gamma_E$, but it is clear that by matching together microlocal solutions, the result does not depend on this simplification.

In fact, BS depend on the geometry of $\gamma_E$ only through its Maslov index, which equals 2 when $\gamma_E$ is a smooth embedded Lagrangian submanifold. The case where $\gamma_E$ is not a submanifold (for example, homeomorphic to the figure-eight) and has Maslov index 0, is not considered here (see~\cite{JO22025} or \cite{YCVPAR1}). (Recall that in dimension 1, the Maslov index, defined modulo 4, is an even number, hence either 0 or 2.)

Our main result is the following:
\begin{thm}
Let $P(x,h D_x;h)$ be a self-adjoint $h$-PDO, given as the Weyl quantization of a real classical symbol
\[p(x,\xi;h)\sim p_0(x,\xi)+h\,p_1(x,\xi)+h^2\,p_2(x,\xi)+\cdots\]
Assume that the geometry of $p_0$ satisfies the hypothesis (H) of Section \ref{S1}, and that $E_{+}<E_{0}=\displaystyle\liminf_{|x,\xi|\rightarrow +\infty}p_{0}(x,\xi)$. Then the spectrum of $P$ in a fixed energy interval $I\subset\mathbb{R}$ is discrete, and given by the BS quantization condition:
\[S_h(E):=S_0(E)+h\,S_1(E)+h^2\,S_2(E)+\cdots=2\pi n h,\quad n\in\mathbb{Z}\]
where:
\begin{itemize}
    \item $S_0(E)=\displaystyle\oint_{\gamma_E}\xi(x)\,dx=\displaystyle\int\int_{\{p_0\leq E\}\cap W}\,d\xi\wedge\,dx$ is the \emph{classical action} along the closed orbit $\gamma_E \subset \{p_0=E\}$;
    \item $S_1(E)=-\pi-\displaystyle\int_{\gamma_E} p_{1}\big(x(t),\xi(t)\big)\,dt$ is the \emph{first-order correction}, including the Maslov index and the integral of the subprincipal 1-form $p_{1}\,dt$;
    \item $S_2(E)$ is the \emph{second-order correction}, given by
    \[S_2(E)={1\over24}\,{d\over dE}\int_{\gamma_E}\Delta\,dt-\int_{\gamma_E}p_{2}\,dt-{1\over2}\,{d\over dE}\int_{\gamma_E}p_1^{2}\,dt\]
    with
    \[\Delta(x,\xi)=\frac{\partial^{2} p_{0}}{\partial x^{2}}\,\frac{\partial^{2} p_{0}}{\partial \xi^{2}}-\big(\frac{\partial^{2} p_{0}}{\partial x\,\partial \xi}\big)^{2},\]
    and $(x(t),\xi(t))$ is a parametrization of $\gamma_E$ by the Hamiltonian flow.
\end{itemize}
\end{thm}
Let us note that the deformation quantization method (see \cite{CdV1}), easily recovers this result, as well as higher-order terms, in particular $S_{4}(E)$ (see \cite{CARGooGR} and \cite{Gracia00Saz} for a diagrammatic approach). Recall that all odd-order terms $S_{j}(E)$ with $j\geq 3$ vanish.
%However, to our knowledge, the method based on the microlocal Wronskian and the Gram matrix is the only one that can be applied to systems such as the Bogoliubov–de Gennes Hamiltonian.
\subsection{Quasi-modes mod $\mathcal{O}(h^{2})$ in Fourier representation}
We first recall H\"{o}rmander's asymptotic stationary phase theorem (see e.g. \cite{LARsHORmander}, Theorem 7.7.5):

Let $\varphi:\mathbb{R}^{d}\rightarrow \mathbb{C}$ be a function such that $\mathrm{Im}\big(\varphi(x)\big)\geq 0$, and suppose that $\varphi$ has a non-degenerate critical point at $x_{0}$. Then, we have the asymptotic expansion:
\begin{equation}\label{Phasestat0}
\int_{\mathbb{R}^{d}}e^{\frac{i}{h}\,\varphi(x)}\,u(x)\,dx \sim e^{\frac{i}{h}\,\varphi(x_{0})}\,\bigg(\det\big(\frac{\varphi''(x_{0})}{2\,i\,\pi\,h}\big)\bigg)^{-\frac{1}{2}}\,\sum_{j}h^{j}\,L_{j} u (x_{0})
\end{equation}
where $L_{j}$ are linear differential operators, with $L_{0}u(x_{0})=u(x_{0})$ and in particular:
\begin{equation}\label{A06}
L_{1}u(x_{0})=\sum_{n=0}^{2}\frac{2^{-(n+1)}}{i\,n!\,(n+1)!}\,\big<\big(\varphi''(x_{0})\big)^{-1}\,D_{x},
D_{x}\big>^{n+1}\,(\Phi_{x_{0}}^{n}\,u)(x_{0})
\end{equation}
with:
\begin{equation}\label{A07}
\Phi_{x_{0}}(x)=\varphi(x)-\varphi(x_{0})-\frac{1}{2}\,\big<\varphi''(x_{0}).\,(x-x_{0}), x-x_{0}\big>
\end{equation}
We note that $\Phi_{x_{0}}$ vanishes to order 3 at $x_{0}$ \big(i.e., $\Phi_{x_{0}}(x_{0})=0,\;\Phi'_{x_{0}}(x_{0})=0,\;\Phi''_{x_{0}}(x_{0})=0$\big).

In the sequel, we present formulas with accuracy up to the second order in $h$. It is helpful to start building the quasi-modes from a focal point, because this gives both outgoing and incoming approximate solutions at the same time.

For \(E\in I\), let \(a_E=(x_E,\xi_E)\in\gamma_E\) be such that \(\left(\displaystyle\frac{\partial p_0}{\partial \xi}\right)(a_E)=0\) (i.e., \(a_E\) is a focal point). Since \(\left(\displaystyle\frac{\partial p_0}{\partial x}\right)(a_E)\neq0\), the orbit \(\gamma_E\) can be locally parametrized near \(a_E\) using a phase function \(\psi(\xi)=\psi(\xi; E)\), which satisfies the Hamilton-Jacobi equation:
\begin{equation}\label{Jacob}
p_0\big(-\psi'(\xi),\xi\big)=E
\end{equation}
and is normalized by \(\psi(\xi_E)=0\). We then look for an asymptotic solution of $\big(P(x,hD_{x};h)-E\big)u(x;h)=0$ of the form
\begin{equation}\label{ua}
u(x;h)=(2\pi h)^{-1/2}\,\int e^{\frac{i}{h}\,x\xi}\,\hat{u}(\xi;h)\,d\xi=(2\pi h)^{-1/2}\,\int e^{\frac{i}{h}\,\big(x\xi+\psi(\xi)\big)}\,b(\xi;h)\,d\xi
\end{equation}
where $\psi,\,b$ depend also on $E$. We aim to compute:
\begin{equation}\label{Puxh}
Pu(x;h)=(2 \pi h)^{-\frac{3}{2}}\,\int\int\int e^{\frac{i}{h}\big(\,(x-y)\,\eta+y\,\xi+\psi(\xi)\,\big)}\,p\big(\frac{x+y}{2},\eta;h\big)\,b(\xi;h)\,d\xi\,dy\,d\eta
\end{equation}
In (\ref{Puxh}), we integrate with respect to the variables \((y,\eta)\). For fixed \(\xi\), the phase is:
$$(y, \eta) \mapsto (x-y)\,\eta+y\,\xi$$
with critical point \((y_c,\eta_c)=(x,\xi)\).  Let us set the change of variables:
$$y-x=2\,y',\quad\eta-\xi=\eta'$$
The Jacobian of this transformation is 2, so:
\begin{equation}\label{}
Pu(x;h)=2\,(2\,\pi\,h)^{-\frac{3}{2}}\,\int\int\int e^{\frac{i}{h}\,\big(-2\,y'\,\eta'+x\,\xi+\psi(\xi)\big)}\,p\big(x+y',\xi+\eta';h\big)\,b(\xi;h)\,d\xi\,dy'\,d\eta'
\end{equation}
Now, for fixed \(\xi\), applying the stationary phase formula to the variables \((y',\eta')\) gives:
\begin{align*}
\int\int e^{-\frac{i}{h} 2 y' \eta'} p(x+y',\xi+\eta';h)\,dy'\,d\eta'& \sim \pi h \sum_{j=0}^{N-1}\frac{1}{j! (2/h)^{j} i^{j}}\big(\partial^{j}_{y'} \partial^{j}_{\eta'} p(x+y',\xi+\eta';h)\big)_{(y',\eta')=(0,0)}\ \\
&\sim \pi h \big(p(x,\xi;h)+\frac{h}{2 i} \frac{\partial^{2} p}{\partial x \partial \xi}(x,\xi;h)-
\frac{h^{2}}{8} \frac{\partial^{4} p}{\partial x^{2} \partial \xi^{2}}(x,\xi;h)+\mathcal{O}(h^{3})\big)
\end{align*}
Hence:
\begin{align*}
Pu(x;h)&=(2 \pi h)^{-\frac{1}{2}} \int e^{\frac{i}{h} \big(x\,\xi+\psi(\xi)\big)} b(\xi;h) \big(p(x,\xi;h)+\frac{h}{2 i} \frac{\partial^{2}
p}{\partial x\,\partial \xi}(x,\xi;h)-
\frac{h^{2}}{8} \frac{\partial^{4} p}{\partial x^{2} \partial \xi^{2}}(x,\xi;h)+\mathcal{O}(h^{3})\big)\,d\xi\ \\
&=(2 \pi h)^{-\frac{1}{2}} \int e^{\frac{i}{h} \big(x\,\xi+\psi(\xi)\big)} b(\xi;h)\,\big(p_{0}(x,\xi)+h \tilde{p}_{1}(x,\xi)+h^{2}
\tilde{p}_{2}(x,\xi)+\mathcal{O}(h^{3})\big)\,d\xi
\end{align*}
where
\begin{equation}\label{}
\tilde{p}_{1}(x,\xi)=p_{1}(x,\xi)+\frac{1}{2i}\,\frac{\partial^{2}\,p_{0}}{\partial x\,\partial \xi}(x,\xi)
\end{equation}
and
\begin{equation}\label{ptilde}
\tilde{p}_{2}(x,\xi)=p_{2}(x,\xi)+\frac{1}{2i}\,\frac{\partial^{2}\,p_{1}}{\partial x\,\partial \xi}(x,\xi)
-\frac{1}{8}\,\frac{\partial^{4}\,p_{0}}{\partial x^{2}\,\partial \xi^{2}}(x,\xi)
\end{equation}
Following \big(\cite{Ycdvl}, the Maslov Ansatz and Theorem 43\big), we look for $b(\xi;h)\sim b_{0}(\xi)+h\,b_{1}(\xi)+\cdots$, a classical elliptic symbol, with $b_{0}(0)\neq 0$, such that there exists a symbol $a(x,\xi;h)\sim a_{0}(x,\xi)+h\,a_{1}(x,\xi)+\cdots$ satisfying:
\begin{equation}\label{E00}
h\,D_{\xi}\big(e^{\frac{i}{h}\,\big(x\,\xi+\psi(\xi)\big)}\,a(x,\xi;h)\big)=e^{\frac{i}{h}\,\big(x\,\xi+\psi(\xi)\big)}\,
b(\xi;h)\,
\big( p_{0}(x,\xi)-E+h\,\tilde{p}_{1}(x,\xi)+h^{2}\,\tilde{p}_{2}(x,\xi)+\mathcal{O}(h^{3})\big)
\end{equation}
or more explicitly
\begin{equation}\label{E0197As880}
\big(x+\psi'(\xi)\big)\,a(x,\xi;h)+hD_{\xi}a(x,\xi;h)=b(\xi;h)\,\big(p_{0}(x,\xi)-E+h\,\tilde{p}_{1}(x,\xi)+h^{2}\,\tilde{p}_{2}(x,\xi)+\mathcal{O}(h^{3})\big)
\end{equation}
In order to solve (\ref{E0197As880}) for $\xi$ near $\xi_{E}$, $x$ near $x_{E}$, $E$ near $E_{0}$, it is sufficient to solve the sequence of equations,
\begin{equation}\label{S045E1}
\big(x+\psi'(\xi)\big)\,a_{0}(x,\xi)=\big(p_{0}(x,\xi)-E\big)\,b_{0}(\xi)
\end{equation}
\begin{equation}\label{S045E2}
\big(x+\psi'(\xi)\big)\,a_{1}(x,\xi)+D_{\xi}\,a_{0}(x,\xi)=\big(p_{0}(x,\xi)-E\big)\,
b_{1}(\xi)+\tilde{p}_{1}(x,\xi)\,b_{0}(\xi)
\end{equation}
\begin{equation}\label{S045E3}
\big(x+\psi'(\xi)\big)\,a_{2}(x,\xi)+D_{\xi}\,a_{1}(x,\xi)=\big(p_{0}(x,\xi)-E\big)\,
b_{2}(\xi)+\tilde{p}_{1}(x,\xi)\,b_{1}(\xi)+\tilde{p}_{2}(x,\xi)\,b_{0}(\xi)
\end{equation}
Here, we have grouped the terms according to the powers of $h$ in equation (\ref{E0197As880}); equation (\ref{S045E1}) is obtained by annihilating the term in $h^0$, equation (\ref{S045E2}) by annihilating the term in $h^1$, and equation (\ref{S045E3}) by annihilating the term in $h^2$. We define the function $\lambda(x,\xi)$ by:
\begin{equation}\label{ETSER174}
\lambda(x,\xi):=\frac{p_{0}(x,\xi)-E}{x+\psi'(\xi)}
\end{equation}
From equation (\ref{ETSER174}), we deduce that:
\begin{equation}\label{ETSER175}
\lambda\big(-\psi'(\xi),\xi\big)=\big(\partial_{x}\,p_{0}\big)\big(-\psi'(\xi),\xi\big):=\alpha(\xi)
\end{equation}
Differentiating both sides of equation (\ref{Jacob}) with respect to $\xi$ gives:
\begin{equation}\label{ETSER176}
\psi''(\xi)=\frac{\big(\partial_{\xi}\,p_{0}\big)\big(-\psi'(\xi),\xi\big)}{\alpha(\xi)}
\end{equation}
which vanishes at $\xi_E$. For a given $b_0$, the unique solution of (\ref{S045E1}) is $a_{0}(x,\xi)=\lambda(x,\xi)\,b_{0}(\xi)$. In order to solve equation (\ref{S045E2}), it is necessary and sufficient that
\begin{equation}
(D_\xi a_0)\big(-\psi'(\xi),\xi\big)=\tilde{p}_{1}\big(-\psi'(\xi),\xi\big)\,b_{0}(\xi)
\end{equation}
This is equivalent to:
\begin{equation}\label{ETSER177}
\big(\partial_{\xi}\,\lambda\big)\big(-\psi'(\xi),\xi\big)\,b_{0}(\xi)+\alpha(\xi)\,b'_{0}(\xi)=\Big(i\,p_{1}\big(-\psi'(\xi),\xi\big)+
\frac{1}{2}\,\big(\frac{\partial^{2} p_{0}}{\partial x\,\partial \xi}\big)\big(-\psi'(\xi),\xi\big)\Big)\,b_{0}(\xi)
\end{equation}
A direct computation from (\ref{ETSER174}) shows that:
\begin{equation}\label{ETSER178}
\big(\partial_{x}\,\lambda\big)\big(-\psi'(\xi),\xi\big)=
\frac{1}{2}\,\big(\frac{\partial^{2} p_{0}}{\partial x^{2}}\big)\big(-\psi'(\xi),\xi\big)
\end{equation}
Differentiating both sides of equation (\ref{ETSER175}) with respect to $\xi$ gives:
\begin{equation}\label{ETSER179}
\big(\partial_{\xi}\,\lambda\big)\big(-\psi'(\xi),\xi\big)=
-\frac{\psi''(\xi)}{2}\,\big(\frac{\partial^{2} p_{0}}{\partial x^{2}}\big)\big(-\psi'(\xi),\xi\big)+
\big(\frac{\partial^{2} p_{0}}{\partial x\,\partial \xi}\big)\big(-\psi'(\xi),\xi\big)
\end{equation}
Substituting this expression into the equation (\ref{ETSER177}) gives the differential equation for $b_0$:
\begin{equation}\label{ETSER180}
\alpha(\xi)\,b'_{0}(\xi)+
\Big(\frac{1}{2}\,\alpha'(\xi)-i\,p_{1}\big(-\psi'(\xi),\xi\big)\,\Big)\,b_{0}(\xi)=0
\end{equation}
whose general solution is:
\begin{equation}\label{ETSER181}
b_{0}(\xi)=C_{0}\,|\alpha(\xi)|^{-\frac{1}{2}}\,\exp\Big(i\,\int\frac{p_{1}\big(-\psi'(\xi),\xi\big)}{\alpha(\xi)}\,d\xi\Big)
\end{equation}
In order to solve (\ref{S045E3}) it is necessary and sufficient that
\begin{equation}\label{ETSER182}
\big(D_{\xi}\,a_{1}\big)\big(-\psi'(\xi),\xi\big)=
\tilde{p}_{1}\big(-\psi'(\xi),\xi\big)\,b_{1}(\xi)+\tilde{p}_{2}\big(-\psi'(\xi),\xi\big)\,b_{0}(\xi)
\end{equation}
From equation (\ref{S045E2}), we get:
\begin{equation}\label{ETSER183}
a_{1}(x,\xi)=\lambda(x,\xi)\,b_{1}(\xi)+\lambda_{0}(x,\xi)
\end{equation}
where
\begin{equation}\label{ETSER184}
\lambda_{0}(x,\xi)=\frac{\tilde{p}_{1}(x,\xi)\,b_{0}(\xi)+i\,\partial_{\xi}\,a_{0}(x,\xi)}{x+\psi'(\xi)}
\end{equation}
Before continuing, we state a lemma we will use.
\begin{lem}\label{L1}
\begin{equation}\label{ETSER185}
\lambda_{0}\big(-\psi'(\xi),\xi\big)=
b_{0}(\xi)\,\big(\partial_{x} p_{1}-\frac{p_{1}}{2\,\alpha}\,\frac{\partial^{2} p_{0}}{\partial x^{2}}\big)_{x=-\psi'(\xi)}-
i\,b_{0}(\xi)\,\big(\frac{\psi''(\xi)}{6}\frac{\partial^{3} p_{0}}{\partial x^{3}}
+\frac{\alpha'}{4\,\alpha}\,\frac{\partial^{2} p_{0}}{\partial x^{2}}\big)_{x=-\psi'(\xi)}
\end{equation}
\begin{equation}\label{ETSER186}
\!\!\!\!\!\!\!\!\!\!\!\big(\frac{\partial \lambda_{0}}{\partial x}\big)\big(-\psi'(\xi),\xi\big)=
\frac{b_{0}(\xi)}{2}\,\big(\frac{\partial^{2} p_{1}}{\partial x^{2}}-
\frac{p_{1}}{3\,\alpha}\,\frac{\partial^{3} p_{0}}{\partial x^{3}}\big)_{x=-\psi'(\xi)}-
i\,\frac{b_{0}(\xi)}{12}\,\big(\frac{\partial^{4} p_{0}}{\partial x^{3}\,\partial \xi}+
\frac{\psi''(\xi)}{2}\frac{\partial^{4} p_{0}}{\partial x^{4}}+
\frac{\alpha'}{\alpha}\,\frac{\partial^{3} p_{0}}{\partial x^{3}}\big)_{x=-\psi'(\xi)}
\end{equation}
\begin{equation}\label{ETSER187}
\big(\frac{\partial^{n}\,\lambda}{\partial\,x^{n}}\big)\big(-\psi'(\xi),\xi\big)=
\frac{1}{n+1}\,(\frac{\partial^{n+1}\,p_{0}}{\partial\,x^{n+1}})(-\psi'(\xi),\xi);\; \forall \,n\in\mathbb{N}
\end{equation}
\end{lem}
Differentiating both sides of equation (\ref{ETSER183}) with respect to $\xi$ and evaluating at $x=-\psi'(\xi)$ gives
\begin{equation}\label{ETSER188}
\big(D_{\xi}\,a_{1}\big)\big(-\psi'(\xi),\xi\big)=
\frac{1}{i}\,\big(\partial_{\xi}\,\lambda\big)\big(-\psi'(\xi),\xi\big)\,b_{1}(\xi)+\frac{1}{i}\,\alpha(\xi)\,b'_{1}(\xi)+
\frac{1}{i}\,\big(\partial_{\xi}\,\lambda_{0}\big)\big(-\psi'(\xi),\xi\big)
\end{equation}
Then, comparing this with equation (\ref{ETSER182}), we see that $b_1$ must satisfy the differential equation
\begin{equation}\label{ETSER189}
\alpha(\xi)\,b'_{1}(\xi)+\Big(\frac{1}{2}\,\alpha'(\xi)-i\,p_{1}\big(-\psi'(\xi),\xi\big)\Big)\,b_{1}(\xi)=
i\,\tilde{p}_{2}\big(-\psi'(\xi),\xi\big)\,b_{0}(\xi)-
\big(\partial_{\xi}\,\lambda_{0}\big)\big(-\psi'(\xi),\xi\big)
\end{equation}
The homogeneous part of this equation is the same as in (\ref{ETSER180}); therefore, we seek a particular solution of the form
\begin{equation}\label{ETSER190}
D_{1}(\xi)\,|\alpha(\xi)|^{-\frac{1}{2}}\,\exp\big(i\,\int_{\xi_{E}}^{\xi}\frac{p_{1}\big(-\psi'(\zeta),\zeta\big)}{\alpha(\zeta)}\,d\zeta\big)
\end{equation}
Using variation of constants, we find
\begin{equation}\label{ETSER191}
\!\!\!\!\!\!\!\!\!\!\!\!\!\!\!\!\!\!\!\!D_{1}(\xi)=\mathrm{sgn}(\alpha(\xi_{E}))\int_{\xi_{E}}^{\xi}|\alpha(\zeta)|^{-\frac{1}{2}}
\Big(ib_{0}(\zeta)\tilde{p}_{2}\big(-\psi'(\zeta),\zeta\big)-\big(\partial_{\zeta}\,\lambda_{0}\big)\big(-\psi'(\zeta),\zeta\big)\Big)\exp\Big(-i \int_{\xi_{E}}^{\zeta}\frac{p_{1}\big(-\psi'(s),s\big)}{\alpha(s)}\,ds\Big)\,d\zeta
\end{equation}
We normalize by setting
\begin{equation}\label{ETSER192}
D_{1}(\xi_E)=0
\end{equation}
So the general solution of the equation is:
\begin{equation}\label{ETSER193}
b_{1}(\xi)=\big(C_{1}+D_{1}(\xi)\big)
\,|\alpha(\xi)|^{-\frac{1}{2}}\,\exp\big(i\,\int_{\xi_{E}}^{\xi}\frac{p_{1}\big(-\psi'(\zeta),\zeta\big)}{\alpha(\zeta)}\,d\zeta\big)
\end{equation}
It follows that:
\begin{equation}\label{ETSER194}
b_{0}(\xi)+h\,b_{1}(\xi)=\big(C_{0}+h\,C_{1}+h\,D_{1}(\xi)\big)
\,|\alpha(\xi)|^{-\frac{1}{2}}\,\exp\big(i\,\int_{\xi_{E}}^{\xi}\frac{p_{1}(-\psi'(\zeta),\zeta)}{\alpha(\zeta)}\,d\zeta\big)
\end{equation}
The integration constants $C_0$ and $C_1=C_1(a_E)$ will be determined by normalizing the microlocal Wronskians as follows
\subsection{Normalisation}\label{NORMALIsation}
We compute the microlocal Wronskian of $(u^{a},\overline{u^{a}})=(u,\overline{u})$ in $\omega^{a}_{\rho}$. Our goal is to normalize the microlocal solution
\[\hat{u}(\xi;h)=e^{\frac{i}{h} \psi(\xi)}\,b(\xi;h)\]
using the microlocal Wronskian. That is, we seek constants \(C_0\) and \(C_1=C_1(a_E)\) such that
\[\mathcal{W}^{a}(\hat{u},\overline{\hat{u}})=1+\mathcal{O}(h^{2}).\]
In the Fourier representation, we write:
$$\mathcal{W}^{a}_{\rho}(\hat{u},\overline{\hat{u}})=\displaystyle \big(\frac{i}{h}\,[P,\chi^{a}]_{\rho}\,\hat{u}|\hat{u}\big)$$
where $\chi^{a}\in C^{\infty}_{0}(\mathbb{R}^{2})$ is a smooth cut-off equal to 1 near the focal point $a=a_E$. Without loss of generality, we can take $\chi^{a}(x,\xi)=\chi_{1}(x)\,\chi_{2}(\xi)$, with $\chi_{2}\equiv 1$ on small neighborhoods $\omega^{a}_{\pm}$, of $\mathrm{supp}\big(\frac{i}{h}\,[P,\chi^{a}]\big)\cap \{p_{0}(x,\xi)=E\}$ in the region $\pm (\xi-\xi_{E})>0$. Therefore, it is sufficient to consider variations of the function \(\chi_1(x)\) only.

In general, if \(P\) and \(Q\) are two \(h\)-pseudodifferential operators, whose Weyl symbols admit the expansions
\begin{equation*}
\sigma^{W}(P)(x,\xi;h)=p_{0}(x,\xi)+h\,p_{1}(x,\xi)+h^{2}\,p_{2}(x,\xi)+\cdots
\end{equation*}
\begin{equation*}
\sigma^{W}(Q)(x,\xi;h)=q_{0}(x,\xi)+h\,q_{1}(x,\xi)+h^{2}\,q_{2}(x,\xi)+\cdots
\end{equation*}
then we have
\begin{equation*}
\sigma^{W}\big(\frac{i}{h}[P,Q]\big)=\{\sigma^{W}(P),\sigma^{W}(Q)\}+\mathcal{O}(h^{2})
\end{equation*}
where \(\{\cdot,\cdot\}\) denotes the Poisson bracket. In particular,
\begin{equation*}
\sigma^{W}\big(\frac{i}{h}\,[P,\chi^{a}]\big)(x,\xi;h):=c(x,\xi;h)=
\big(\partial_{\xi}\,p_{0}(x,\xi)+h\,\partial_{\xi}\,p_{1}(x,\xi)\big)\,\chi'_{1}(x)+\mathcal{O}(h^{2})
\end{equation*}
Using the Weyl calculus, the operator acts in Fourier representation as
\begin{equation*}
c^{w}(-h\,D_{\xi},\xi;h)\,v(\xi;h)=(2\,\pi\,h)^{-1}\,\int\int e^{-\frac{i}{h}\,(\xi-\eta)\,y}\,c\big(y,\frac{\xi+\eta}{2};h\big)\,v(\eta;h)\,dy\,d\eta
\end{equation*}
Hence, applying to $\hat{u}$, we get
\begin{equation*}
\displaystyle \frac{i}{h}\,[P,\chi^{a}]\,\hat{u}(\xi;h)=(2\,\pi\,h)^{-1}\,\int\int e^{\frac{i}{h}\big(\psi(\eta)-(\xi-\eta)\,y\big)}\,c\big(y,\frac{\xi+\eta}{2};h\big)\,b(\eta;h)\,dy\,d\eta
\end{equation*}
For fixed $\xi$, the phase function corresponding to the oscillatory integral defining $\displaystyle \frac{i}{h}\,[P,\chi^{a}]\,\hat{u}$ is given by
\begin{equation*}
\varphi_{\xi}(y,\eta)=\psi(\eta)-(\xi-\eta)\,y
\end{equation*}
The critical points of $\varphi_{\xi}$ are
\begin{equation*}
\big(y_{c}(\xi),\eta_{c}(\xi)\big)=\big(-\psi'(\xi),\xi\big),
\end{equation*}
and therefore, the corresponding critical values of $\varphi_{\xi}$ are
\begin{equation*}
\varphi_{\xi}\big(y_{c}(\xi),\eta_{c}(\xi)\big)=\varphi_{\xi}\big(-\psi'(\xi),\xi\big)=\psi(\xi)
\end{equation*}
A direct computation shows that
\begin{equation*}
\big(\mathrm{Hess}\,\varphi_{\xi}\big)\big(y_{c}(\xi),\eta_{c}(\xi)\big)=
\left(
\begin{array}{cc}
0 & 1 \\
1 & \psi''(\xi) \\
\end{array}
\right)
\end{equation*}
Let $c_{j}(y,\eta):=\partial_{\eta} p_{j}(y,\eta)\,\chi'_{1}(y)$, for all $j\in\{0,1\}$. By the stationary phase theorem (\ref{Phasestat0}), we obtain:
\begin{equation*}
\frac{i}{h}\,[P,\chi^{a}]\,\hat{u}(\xi;h)=e^{\frac{i}{h}\,\psi(\xi)}\,\big(d_{0}(\xi)+h\,d_{1}(\xi)+\mathcal{O}(h^{2})\big)
\end{equation*}
with
\begin{equation*}
d_{0}(\xi)=c_{0}\big(-\psi'(\xi),\xi\big)\,b_{0}(\xi)
\end{equation*}
and
\begin{equation*}
d_{1}(\xi)=c_{0}\big(-\psi'(\xi),\xi\big)\,b_{1}(\xi)+
c_{1}\big(-\psi'(\xi),\xi\big)\,b_{0}(\xi)+\frac{i}{2}\,J(\xi)
\end{equation*}
where
\begin{equation*}
J(\xi)=e_{0}'(\xi)\,b_{0}(\xi)+2\,e_{0}(\xi)\,b'_{0}(\xi)
\end{equation*}
and where we have set
\begin{equation*}
e_{0}(\xi)=\partial_{x} c_{0}\big(-\psi'(\xi),\xi\big)
\end{equation*}
It follows that
\begin{align*}
\mathcal{W}^{a}_{+}(\hat{u},\overline{\hat{u}})&=
\int_{\xi_{E}}^{+\infty}d_{0}(\xi)\,\overline{b_{0}(\xi)}\,d\xi+h\,
\int_{\xi_{E}}^{+\infty}\big(d_{0}(\xi)\,\overline{b_{1}(\xi)}+
d_{1}(\xi)\,\overline{b_{0}(\xi)}\big)\,d\xi+\mathcal{O}(h^{2})\ \\
&=\mathcal{M}_{0}^{+}+h\,\mathcal{M}_{1}^{+}+\mathcal{O}(h^{2})
\end{align*}
First, we have
\begin{align*}
\mathcal{M}_{0}^{+}&=
\int_{\xi_{E}}^{+\infty}\partial_{\xi} p_{0}\big(-\psi'(\xi),\xi\big)
\,\chi'_{1}\big(-\psi'(\xi)\big)\,|b_{0}(\xi)|^{2}\,d\xi\ \\
&=
|C_{0}|^{2}\,\int_{\xi_{E}}^{+\infty}\frac{\alpha(\xi)}{|\alpha(\xi)|}\,
\psi''(\xi)\,\,\chi'_{1}\big(-\psi'(\xi)\big)\,d\xi\ \\
&=-|C_{0}|^{2}\,\int_{\xi_{E}}^{+\infty}\mathrm{sgn}(\alpha(\xi))\,
\frac{d}{d\xi}\big(\chi_{1}\big(-\psi'(\xi)\big)\big)\,d\xi\ \\
&=|C_{0}|^{2}\,\mathrm{sgn}\big(\alpha(\xi_{E})\big)
\end{align*}
The next step is to compute
\begin{equation}\label{devient}
\mathcal{M}_{1}^{+}:=\int_{\xi_{E}}^{+\infty}\big(d_{0}(\xi)\,\overline{b_{1}(\xi)}+
d_{1}(\xi)\,\overline{b_{0}(\xi)}\big)\,d\xi
\end{equation}
A few lines of calculations show that
\begin{align*}
d_{0}\,\overline{b_{1}}+d_{1}\,\overline{b_{0}}&=
-2\,\mathrm{Re}(\overline{C_{0}}\,C_{1})\,\frac{d}{d\xi}(\chi_{1})\,\mathrm{sgn}(\alpha)
-2\,\mathrm{Re}(\overline{C_{0}}\,D_{1})\,\frac{d}{d\xi}(\chi_{1})\,\mathrm{sgn}(\alpha)\ \\
&+\frac{|C_{0}|^{2}}{|\alpha|}\,\big(\partial_{\xi} p_{1}\,\chi'_{1}-s_{0}\,p_{1}\big)
+\frac{i}{2}\,|C_{0}|^{2}\,\mathrm{sgn}(\alpha)\,s'_{0}
\end{align*}
where we set
$$s_{0}(\xi)=\frac{e_{0}(\xi)}{\alpha(\xi)}$$
and therefore
\begin{align*}
\mathcal{M}_{1}^{+}&=-2\,\mathrm{Re}(\overline{C_{0}}\,C_{1})\,\int_{\xi_{E}}^{+\infty}\,
\mathrm{sgn}(\alpha)\,\frac{d}{d\xi}(\chi_{1})\,d\xi
-2\,\mathrm{Re}\,\big(\overline{C_{0}}\,\int_{\xi_{E}}^{+\infty}
\mathrm{sgn}(\alpha)\,D_{1}(\xi)\,\frac{d}{d\xi}(\chi_{1})\,\,d\xi\big) \ \\ &+|C_{0}|^{2}\,\int_{\xi_{E}}^{+\infty}\frac{1}{|\alpha|}\,\big(\partial_{\xi} p_{1}\,\chi'_{1}-s_{0}\,p_{1}\big)\,d\xi+
\frac{i}{2}\,|C_{0}|^{2}\,
\int_{\xi_{E}}^{+\infty}\mathrm{sgn}(\alpha)\,s'_{0}\,d\xi
\end{align*}
Note that
\begin{equation*}
\partial_{x} c_{0}(x,\xi)=
\frac{\partial^{2} p_{0}(x,\xi)}{\partial x\,\partial \xi}\,\chi'_{1}(x)+
\partial_{\xi} p_{0}(x,\xi)\,\chi''_{1}(x)
\end{equation*}
and
\begin{equation*}
\chi'_{1}(x_{E})=0,\;\chi''_{1}(x_{E})=0,\;\displaystyle \lim_{\xi\rightarrow +\infty}\chi'_{1}\big(-\psi'(\xi)\big)=0,\;\displaystyle \lim_{\xi\rightarrow +\infty}\chi''_{1}\big(-\psi'(\xi)\big)=0,\;\alpha(\xi_{E})\neq 0
\end{equation*}
This implies that
\begin{align*}
\int_{\xi_{E}}^{+\infty}\mathrm{sgn}\big(\alpha(\xi)\big)\,s'_{0}(\xi)\,d\xi&=
\int_{\xi_{E}}^{+\infty}\mathrm{sgn}\big(\alpha(\xi)\big)\,
\frac{d}{d\xi}\big(\frac{\tilde{s}_{0}(\xi)}{\alpha(\xi)}\big)\,
d\xi\ \\
&=\int_{\xi_{E}}^{+\infty}\frac{d}{d\xi}\big(\frac{\tilde{s}_{0}(\xi)}{|\alpha(\xi)|}\big)\,d\xi\ \\
&=\big[\frac{\tilde{s}_{0}(\xi)}{|\alpha(\xi)|}\big]_{\xi_{E}}^{+\infty}=
\big[\frac{\partial_{x} c_{0}\big(-\psi'(\xi),\xi\big)}{|\alpha(\xi)|}\,\big]_{\xi_{E}}^{+\infty}=0
\end{align*}
and
\begin{equation*}
\!\!\!\int_{\xi_{E}}^{+\infty}\,
\mathrm{sgn}\big(\alpha(\xi)\big)\,\frac{d}{d\xi}\bigg(\chi_{1}\big(-\psi'(\xi)\big)\bigg)\,d\xi=
\big[\mathrm{sgn}\big(\alpha(\xi)\big)\,\chi_{1}\big(-\psi'(\xi)\big)\big]_{\xi_{E}}^{+\infty}
=-\mathrm{sgn}\big(\alpha(\xi_{E})\big)\,\chi_{1}(x_{E})=-\mathrm{sgn}\big(\alpha(\xi_{E})\big)
\end{equation*}
So equation (\ref{devient}) becomes:
\begin{equation}\label{mathcalS1}
\mathcal{M}_{1}^{+}
=2\,\mathrm{sgn}\big(\alpha(\xi_{E})\big)\,\mathrm{Re}(\overline{C_{0}}\,C_{1})-2\,I_{1}+I_{2}
\end{equation}
where we set
\begin{equation}\label{i1}
I_{1}=\mathrm{Re}\Big(\overline{C_{0}}\,
\int_{\xi_{E}}^{+\infty}\mathrm{sgn}\big(\alpha(\xi)\big)\,D_{1}(\xi)\,
\frac{d}{d\xi}\big(\chi_{1}\big(-\psi'(\xi)\big)\big)\,\,d\xi\Big)
\end{equation}
and
\begin{equation}\label{i2}
I_{2}=|C_{0}|^{2}\,\int_{\xi_{E}}^{+\infty}\frac{1}{|\alpha(\xi)|}\,
\Big(\partial_{\xi} p_{1}\big(-\psi'(\xi),\xi\big)\,
\chi'_{1}\big(-\psi'(\xi)\big)-s_{0}(\xi)\,p_{1}\big(-\psi'(\xi),\xi\big)\Big)\,d\xi
\end{equation}
After a few integrations by parts, we obtain that
\begin{equation}\label{}
I_{1}=-\frac{|C_{0}|^{2}}{2}\,\mathrm{sgn}\big(\alpha(\xi_{E})\big)\,
\partial_{x}\big(\displaystyle \frac{p_{1}}{\partial_{x} p_{0}}\big)(a_{E})+\frac{|C_{0}|^{2}}{2}\,
\int_{\xi_{E}}^{+\infty}\frac{\psi''\,\chi'_{1}}{|\alpha|}\,\big(\partial_{x} p_{1}-\frac{p_{1}}{\alpha}\,
\frac{\partial^{2} p_{0}}{\partial x^{2}}\big)\,d\xi
\end{equation}
and
\begin{equation}\label{}
I_{2}=|C_{0}|^{2}\,
\int_{\xi_{E}}^{+\infty}\frac{\psi''\,\chi'_{1}}{|\alpha|}\,\big(\partial_{x} p_{1}-\frac{p_{1}}{\alpha}\,
\frac{\partial^{2} p_{0}}{\partial x^{2}}\big)\,d\xi
\end{equation}
Finally
\begin{equation*}
\mathcal{M}_{1}^{+}=2\,\mathrm{sgn}\big(\alpha(\xi_{E})\big)\,\mathrm{Re}(\overline{C_{0}}\,C_{1})+|C_{0}|^{2}\,\mathrm{sgn}\big(\alpha(\xi_{E})\big)\,
\partial_{x}\big(\displaystyle \frac{p_{1}}{\partial_{x} p_{0}}\big)(a_{E})
\end{equation*}
We thus have modulo $\mathcal{O}(h^{2})$
\begin{equation}\label{}
\mathcal{W}^{a}_{+}(\hat{u},\overline{\hat{u}})=|C_{0}|^{2}\,\mathrm{sgn}\big(\alpha(\xi_{E})\big)+
h\,\mathrm{sgn}\big(\alpha(\xi_{E})\big)\,\Big(2\,\mathrm{Re}(\overline{C_{0}}\,C_{1})+|C_{0}|^{2}\,
\partial_{x}\big(\displaystyle \frac{p_{1}}{\partial_{x} p_{0}}\big)(a_{E})\Big)
\end{equation}
Similarly, one shows that modulo $\mathcal{O}(h^{2})$
\begin{equation}\label{}
\mathcal{W}^{a}_{-}(\hat{u},\overline{\hat{u}})=-|C_{0}|^{2}\,
\mathrm{sgn}\big(\alpha(\xi_{E})\big)-
h\,\mathrm{sgn}\big(\alpha(\xi_{E})\big)\,\Big(2\,\mathrm{Re}(\overline{C_{0}}\,C_{1})+|C_{0}|^{2}\,
\partial_{x}\big(\displaystyle \frac{p_{1}}{\partial_{x} p_{0}}\big)(a_{E})\Big)
\end{equation}
which allows us to conclude that modulo $\mathcal{O}(h^{2})$
\begin{equation}\label{}
\mathcal{W}^{a}(\hat{u},\overline{\hat{u}}):=
\mathcal{W}^{a}_{+}(\hat{u},\overline{\hat{u}})-\mathcal{W}^{a}_{-}(\hat{u},\overline{\hat{u}})=2\,|C_{0}|^{2}\,\mathrm{sgn}\big(\alpha(\xi_{E})\big)+2\,
h\,\mathrm{sgn}\big(\alpha(\xi_{E})\big)\,\Big(2\,\mathrm{Re}(\overline{C_{0}}\,C_{1})+|C_{0}|^{2}\,
\partial_{x}\big(\displaystyle \frac{p_{1}}{\partial_{x} p_{0}}\big)(a_{E})\Big)
\end{equation}
Assuming $\alpha(\xi_{E})>0$, $C_{0}>0$, and $C_{1}\in \mathbb{R}$, it follows that
\begin{equation}\label{C00}
C_{0}=2^{-1/2}
\end{equation}
and
\begin{equation}\label{C01}
C_{1}:=C_{1}(a_{E})=-2^{-3/2}\,\partial_{x}\big(\frac{p_{1}}{\partial_{x}\,p_{0}}\big)(a_{E})
\end{equation}
For the values of $C_{0}$ and $C_{1}$ found above, we indeed have
\begin{equation}\label{}
\mathcal{W}^{a}(\hat{u},\overline{\hat{u}})=1+\mathcal{O}(h^{2})
\end{equation}
We say that $u^{a}$ is well-normalized mod $\mathcal{O}(h^{2})$. This can be formalized by considering $\{a_{E}\}$ as a Poincaré section (see Section \ref{ACTIO0NANGLE}), and \emph{Poisson operator} the operator that assigns, in a unique way, to the initial condition $C_{0}$ on $\{a_{E}\}$ the well normalized (forward) solution $u^{a}$ to $(P-E)u^{a}=0$: namely, $C_{1}(E)$ and $D_{1}(\xi)$, hence also $\widehat{u^{a}}$, depend linearly on $C_{0}$.
\begin{rem}\label{}
So far, under the assumption $\alpha(\xi_{E})>0$, we have obtained the following expression
\begin{equation}\label{}
\widehat{u^{a}}(\xi;h)=\big(C_{0}+h\,C_{1}(a_{E})+h\,D_{1}(\xi)+\mathcal{O}(h^{2})\big)\,|\alpha(\xi)|^{-\frac{1}{2}}\,
\exp\Big[\frac{i}{h}\,\big(\psi(\xi)+h\,\int_{\xi_{E}}^{\xi}\frac{p_{1}(-\psi'(\zeta),\zeta)}{\alpha(\zeta)}\,d\zeta\,\big)\Big]
\end{equation}
Thanks to the identity
\begin{equation}\label{}
C_{0}+h\,C_{1}(a_{E})+h\,D_{1}(\xi)=
\Big(C_{0}+h\,C_{1}+h\,\mathrm{Re}\big(D_{1}(\xi)\big)\Big)\,
\exp\Big[\frac{i\,h}{C_{0}}\,\mathrm{Im}\big(D_{1}(\xi)\big)\Big]+\mathcal{O}(h^{2})
\end{equation}
we can refine both the phase and the half-density up to the next order as follows
\begin{equation}\label{UCHAPEAU}
\widehat{u^{a}}(\xi;h)=\Big(C_{0}+h\,C_{1}(a_{E})+h\,\mathrm{Re}\big(D_{1}(\xi)\big)\Big)\,|\alpha(\xi)|^{-\frac{1}{2}}\,
\exp\big[\frac{i}{h}\,\tilde{S}(\xi,\xi_{E};h)\big]\,\big(1+\mathcal{O}(h^{2})\big)
\end{equation}
where the improved phase \(\tilde{S}\) is defined by
\begin{equation}\label{}
\tilde{S}(\xi,\xi_{E};h)=\psi(\xi)+h\,\int_{\xi_{E}}^{\xi}\frac{p_{1}\big(-\psi'(\zeta),\zeta\big)}{\alpha(\zeta)}\,d\zeta+
\frac{h^{2}}{C_{0}}\mathrm{Im}\big(D_{1}(\xi)\big)
\end{equation}
\end{rem}
\subsection{The homology class of the generalized action: Fourier representation}
Here we identify the various term in (\ref{UCHAPEAU}). First on a $\gamma_{E}$ (i.e. $\Lambda_{E}$) we have
\begin{equation*}
\psi(\xi)=\int-x(\xi)\,d\xi+\mathrm{Const}
\end{equation*}
and
\begin{equation*}
\varphi(x)=\int \xi(x)\,dx+\mathrm{Const}
\end{equation*}
By Hamilton equations
$$\left\{
\begin{array}{ll}
\dot{\xi}(t)&=-\partial_{x} p_{0}\big(x(t),\xi(t)\big)\\
\dot{x}(t)&=\partial_{\xi} p_{0}\big(x(t),\xi(t)\big)
\end{array}
\right.$$
so
\begin{equation*}
\int \frac{p_{1}\big(-\psi'(\xi),\xi\big)}{\alpha(\xi)}\,d\xi=-\int_{\gamma_{E}} p_{1}\big(x(t),\xi(t)\big)\,dt
\end{equation*}
where $t$ is the parametrization of $\gamma_{E}$ by the time evolution. The form $p_{1}\,dt$ is called sub-principal 1-form. Next we consider $D_{1}(\xi)$ as the integral over $\gamma_{E}$ of the 1-form $\Omega_{1}(\xi)$, defined near $\xi_{E}$ in Fourier representation by the expression (\ref{ETSER191}). Using WKB construction, $\Omega_{1}(\xi)$ can also be extended to the spatial representation. Since $\gamma_{E}$ is Lagrangian, $\Omega_{1}(\xi)$ is clearly closed on $\gamma_{E}$; our goal is to compute it modulo exact forms. Using integration by parts in (\ref{ETSER191}), together with the condition $D_{1}(\xi_{E})=0$, we derive the following relations
\begin{align}
\sqrt{2}\,\mathrm{Re}\big(D_{1}(\xi)\big)&=\frac{1}{2}\bigg[\frac{p_{1}\big(-\psi'(\zeta),\zeta\big)\,\big(\displaystyle\frac{\partial^{2} p_{0}}{\partial x^{2}}\big)\big(-\psi'(\zeta),\zeta\big)
-\alpha(\zeta)\,\big(\partial_{x} p_{1}\big)\big(-\psi'(\zeta),\zeta\big)}{\alpha^{2}(\zeta)}\bigg]^{\xi}_{\xi_{E}}\nonumber\ \\
&=-\frac{1}{2}\,\bigg[\partial_{x}\big(\frac{p_{1}}{\partial_{x} p_{0}}\big)\big(-\psi'(\zeta),\zeta\big)\bigg]^{\xi}_{\xi_{E}}\nonumber\ \\
&=-\frac{1}{2}\,\partial_{x}\big(\frac{p_{1}}{\partial_{x} p_{0}}\big)\big(-\psi'(\xi),\xi\big)+
\frac{1}{2}\,\partial_{x}\big(\frac{p_{1}}{\partial_{x} p_{0}}\big)(a_{E})\nonumber\ \\
&=-\frac{1}{2}\,\partial_{x}\big(\frac{p_{1}}{\partial_{x} p_{0}}\big)\big(-\psi'(\xi),\xi\big)-\sqrt{2}\,C_{1}(a_{E})\label{Real}
\end{align}
and
\begin{equation}\label{Racin2ImDind1}
\sqrt{2}\,\mathrm{Im}\big(D_{1}(\xi)\big)
=\int_{\xi_{E}}^{\xi}T_{1}(\zeta)\,d\zeta+\big[\frac{\psi''}{6\,\alpha}\,\frac{\partial^{3} p_{0}}{\partial x^{3}}+\frac{\alpha'}{4\,\alpha^{2}}\, \frac{\partial^{2} p_{0}}{\partial x^{2}}\big]_{\xi_{E}}^{\xi}
\end{equation}
with
\begin{equation}\label{TINdice1}
T_{1}(\zeta)=\frac{1}{\alpha}\,\Big(p_{2}-\frac{1}{8}\,\frac{\partial^{4} p_{0}}{\partial x^{2} \partial \zeta^{2}}+\frac{\psi''}{12}\,\frac{\partial^{4} p_{0}}{\partial x^{3} \partial \zeta}+\frac{(\psi'')^{2}}{24}\,\frac{\partial^{4} p_{0}}{\partial x^{4}}\Big)+
\frac{1}{8}\,\frac{(\alpha')^{2}}{\alpha^{3}}\,\frac{\partial^{2} p_{0}}{\partial x^{2}} +\frac{1}{6}\,\psi''\,\frac{\alpha'}{\alpha^{2}}\,\frac{\partial^{3} p_{0}}{\partial x^{3}}-\frac{p_{1}}{\alpha^{2}}\,\Big(\partial_{x} p_{1}
-\frac{p_{1}}{2\,\alpha}\,\frac{\partial^{2} p_{0}}{\partial x^{2}}\Big)
\end{equation}
There follows:
\begin{lem}\label{LemmMa10101}
Modulo the integral of an exact form in $\mathcal{A}$, with $T_{1}$ as in (\ref{TINdice1}) we have:
\begin{equation}
\begin{aligned}
\mathrm{Re}\big(D_{1}(\xi)\big)&\equiv0\ \\
\sqrt{2}\,\mathrm{Im}\big(D_{1}(\xi)\big)&\equiv\int_{\xi_{E}}^{\xi}T_{1}(\zeta)\,d\zeta
\end{aligned}
\end{equation}
\end{lem}
Passing from Fourier to spatial representation, we can carry the integration in $x$-variable between the focal points $a_{E}$ and $a'_{E}$ , and in $\xi$-variable again near $a'_{E}$. Since $\gamma_{E}$ is smoothly embedded, the microlocal solution $\widehat{u^{a}}$
extends uniquely along $\gamma_{E}$.

Let $u(x,\xi)$ and $v(x,\xi)$ be twoo smooth functions on $\mathcal{A}$, and define the 1-form $\Omega(x,\xi)=u(x,\xi)\,dx+v(x,\xi)\,d\xi$. By Stokes' formula, we have:
\begin{equation*}
\int_{\gamma_{E}}\Omega(x,\xi)=\int\int_{\{p_{0}\leq E\}}\big(\partial_{x} v-\partial_{\xi} u\big)\,dx \wedge{d\xi}
\end{equation*}
According to \cite{CdV1}, we can extend $p_{0}$ inside the disk bounded by $A_{-}$ (which, without loss of generality, may be assumed to contain the origin), so that it coincides with a harmonic oscillator in a neighborhood of a point inside, say $p_{0}(0,0)=0$. Making the symplectic change of coordinates $(x,\xi)\mapsto (t,E)$ in $T^{*}\mathbb{R}$
\begin{equation}\label{St2}
\int\int_{\{p_{0}\leq E\}}\big(\partial_{x} v-\partial_{\xi} u\big)\,dx \wedge{d\xi}=\int_{0}^{E}\int_{0}^{T(E')}\big(\partial_{x} v-\partial_{\xi} u\big)\,dt \wedge{dE'}
\end{equation}
where $T(E')$ is the period of the flow of Hamilton vector field $H_{p_{0}}$ at energy $E'$ \big($T(E')$ being a constant near (0,0)\big). Taking derivative with respect to $E$, we find:
\begin{equation}\label{FORMULASTOKESS}
\frac{d}{dE}\,\int_{\gamma_{E}}\Omega(x,\xi)=\int_{0}^{T(E)}\big(\partial_{x} v-\partial_{\xi} u\big)\,dt
\end{equation}
We compute $\displaystyle\int_{\xi_{E}}^{\xi}T_{1}(\zeta)\,d\zeta$ with $T_{1}$ as in (\ref{TINdice1}), and start to simplify $J_{1}=\displaystyle \int\omega_{1}$, where the 1-form $\omega_{1}$ is expressed in terms of the variable $\xi$ as follows:
\begin{align*}
\omega_{1}\big(-\psi'(\xi),\xi\big)&=\displaystyle \frac{p_{1}\big(-\psi'(\xi),\xi\big)}{\alpha^{2}(\xi)}\,\bigg(\partial_{x} p_{1}\big(-\psi'(\xi),\xi\big)-\frac{p_{1}\big(-\psi'(\xi),\xi\big)}{2\,\alpha(\xi)}\,\displaystyle \frac{\partial^{2} p_{0}}{\partial x^{2}}\big(-\psi'(\xi),\xi\big)\bigg)\;d\xi
\end{align*}
Let $$f_{1}(x,\xi):=\frac{p_{1}^{2}(x,\xi)}{\partial_{x} p_{0}(x,\xi)},$$
then its partial derivative with respect to $x$ is given by:
$$\partial_{x} f_{1}(x,\xi)=\displaystyle \frac{2\,p_{1}(x,\xi)}{\partial_{x} p_{0}(x,\xi)}\,\bigg(\partial_{x} p_{1}(x,\xi)-\frac{p_{1}(x,\xi)}{2\,
\partial_{x} p_{0}(x,\xi)}
\,\displaystyle \frac{\partial^{2} p_{0}(x,\xi)}{\partial x^{2}}\bigg)$$
By (\ref{FORMULASTOKESS}) we get
\begin{equation}\label{GRonDJinD1}
\begin{aligned}
J_{1}&=\frac{1}{2}\,\int_{\gamma_{E}}\frac{\partial_{x}\,f_{1}(x,\xi)}{\partial_{x}\,p_{0}(x,\xi)}\,d\xi=-\frac{1}{2}\,\int_{0}^{T(E)}\partial_{x} f_{1}\big(x(t),\xi(t)\big)\,dt\ \\
     &=-\frac{1}{2}\,\frac{d}{dE}\int_{\gamma_{E}}f_{1}(x,\xi)\,d\xi=-\frac{1}{2}\,\frac{d}{dE}\int_{\gamma_{E}}\frac{p_{1}^{2}(x,\xi)}{\partial_{x} p_{0}(x,\xi)}\,d\xi\ \\
     &=\frac{1}{2}\,\frac{d}{dE}\int_{0}^{T(E)}p_{1}^{2}\big(x(t),\xi(t)\big)\,dt
\end{aligned}
\end{equation}
which is the contribution of $p_{1}$ to the second term $S_{2}$ of generalized action in \big(\cite{CdV1}, Thm2\big). Here $T(E)$ is the period on $\gamma_{E}$. We also have
\begin{equation}\label{}
\int_{\xi_{E}}^{\xi} \frac{1}{\alpha(\zeta)}\,p_{2}\big(-\psi'(\zeta),\zeta\big)\,d\zeta=
\int_{\gamma_{E}}\frac{p_{2}(x,\xi)}{\partial_{x}p_{0}(x,\xi)}\,d\xi=-\int_{0}^{T(E)}p_{2}\big(x(t),\xi(t)\big)\,dt
\end{equation}
In order to compute $T_{1}$ modulo exact forms, it remains to simplify in equation (\ref{TINdice1}) the expression
\begin{align*}
J_{2}&=\int_{\xi_{E}}^{\xi} \frac{1}{\alpha}\,\Big(-\frac{1}{8}\, \frac{\partial^{4} p_{0}}{\partial x^{2}\,\partial \zeta^{2}}+\frac{\psi''}{12}\,\frac{\partial^{4} p_{0}}{\partial x^{3}\,\partial \zeta}+
\frac{(\psi'')^{2}}{24}\,\frac{\partial^{4} p_{0}}{\partial x^{4}}\Big)\,d\zeta+
\frac{1}{8}\,\int_{\xi_{E}}^{\xi} \frac{(\alpha')^{2}}{\alpha^{3}}\, \frac{\partial^{2} p_{0}}{\partial x^{2}}\,d\zeta\ \\
&+\frac{1}{6}\,\int_{\xi_{E}}^{\xi}\,\psi''\,\frac{\alpha'}{\alpha^{2}}\, \frac{\partial^{3} p_{0}}{\partial x^{3}}\,d\zeta+[\frac{\psi''}{6\,\alpha}\, \frac{\partial^{3} p_{0}}{\partial x^{3}}+\frac{\alpha'}{4\,\alpha^{2}}\, \frac{\partial^{2} p_{0}}{\partial x^{2}}]_{\xi_{E}}^{\xi}
\end{align*}
Let
\begin{equation*}
f_{0}(x,\xi):=\frac{\Delta(x,\xi)}{\partial_{x} p_{0}(x,\xi)},
\end{equation*}
where we have set according to \cite{CdV1}
$$\Delta(x,\xi)=\frac{\partial^{2} p_{0}}{\partial x^{2}}\,\frac{\partial^{2} p_{0}}{\partial \xi^{2}}-\big(\frac{\partial^{2} p_{0}}{\partial x\,\partial \xi}\big)^{2}$$
From the eikonal equation (\ref{Jacob}), we deduce
\begin{align*}
\big(\frac{\partial^{2} p_{0}}{\partial \zeta^{2}}\big)\big(-\psi'(\zeta),\zeta\big)&=
\psi'''(\zeta)\,\alpha(\zeta)+\psi''(\zeta)\,\alpha'(\zeta)+
\psi''(\zeta)\,\big(\frac{\partial^{2} p_{0}}{\partial x\,\partial \zeta}\big)\big(-\psi'(\zeta),\zeta\big)\ \\
&=\psi'''(\zeta)\,\alpha(\zeta)+2\,\psi''(\zeta)\,\alpha'(\zeta)+
\big(\psi''(\zeta)\big)^{2}\,\big(\frac{\partial^{2} p_{0}}{\partial x^{2}}\big)\big(-\psi'(\zeta),\zeta\big)
\end{align*}
Consequently,
\begin{equation*}
\Delta\big(-\psi'(\zeta),\zeta\big)=-\big(\alpha'(\zeta)\big)^{2}+\psi'''(\zeta)\,\alpha(\zeta)\,
\big(\frac{\partial^{2} p_{0}}{\partial x^{2}}\big)\big(-\psi'(\zeta),\zeta\big)
\end{equation*}
and
\begin{equation*}
\big(\partial_{x} \Delta\big)\big(-\psi'(\zeta),\zeta\big)=\frac{\partial^{3} p_{0}}{\partial x^{3}}\,\big(\psi'''\,\alpha+2\,\psi''\,\alpha'\big)+
\frac{\partial^{2} p_{0}}{\partial x^{2}}\,\big(\alpha''+\psi'''\,\frac{\partial^{2} p_{0}}{\partial x^{2}}\big)-
2\,\alpha'\,\frac{\partial^{3} p_{0}}{\partial x^{2}\,\partial \zeta}
\end{equation*}
which implies
\begin{equation*}
\frac{(\partial_{x} f_{0})\big(-\psi'(\zeta),\zeta\big)}{\alpha(\zeta)}=
\frac{\psi'''}{\alpha}\,\frac{\partial^{3} p_{0}}{\partial x^{3}}+
2\,\psi''\,\frac{\alpha'}{\alpha^{2}}\,\frac{\partial^{3} p_{0}}{\partial x^{3}}+
\frac{\alpha''}{\alpha^{2}}\,\frac{\partial^{2} p_{0}}{\partial x^{2}}-
2\,\frac{\alpha'}{\alpha^{2}}\,\frac{\partial^{3} p_{0}}{\partial x^{2}\,\partial \zeta}
+\frac{(\alpha')^{2}}{\alpha^{3}}\,\frac{\partial^{2} p_{0}}{\partial x^{2}}
\end{equation*}
By integration by parts, we obtain:
\begin{equation*}
\int_{\xi_{E}}^{\xi}\frac{\psi'''}{\alpha}\,\frac{\partial^{3} p_{0}}{\partial x^{3}}\,d\zeta=
\big[\frac{\psi''}{\alpha}\,\frac{\partial^{3} p_{0}}{\partial x^{3}}\big]_{\xi_{E}}^{\xi}+
\int_{\xi_{E}}^{\xi}\psi''\,\frac{\alpha'}{\alpha^{2}}\,\frac{\partial^{3} p_{0}}{\partial x^{3}}\,d\zeta+
\int_{\xi_{E}}^{\xi}\frac{(\psi'')^{2}}{\alpha}\,
\frac{\partial^{4} p_{0}}{\partial x^{4}}\,d\zeta-
\int_{\xi_{E}}^{\xi}\frac{\psi''}{\alpha}\,\frac{\partial^{4} p_{0}}{\partial x^{3}\,\partial \zeta}\,d\zeta
\end{equation*}
\begin{equation*}
\int_{\xi_{E}}^{\xi}\frac{\alpha''}{\alpha^{2}}\,\frac{\partial^{2} p_{0}}{\partial x^{2}}\,d\zeta=
\big[\frac{\alpha'}{\alpha^{2}}\,\frac{\partial^{2} p_{0}}{\partial x^{2}}\big]_{\xi_{E}}^{\xi}+
2\,\int_{\xi_{E}}^{\xi}\frac{(\alpha')^{2}}{\alpha^{3}}\,\frac{\partial^{2} p_{0}}{\partial x^{2}}\,d\zeta+
\int_{\xi_{E}}^{\xi}\psi''\frac{\alpha'}{\alpha^{2}}\,
\frac{\partial^{3} p_{0}}{\partial x^{3}}\,d\zeta-
\int_{\xi_{E}}^{\xi}\frac{\alpha'}{\alpha^{2}}\,\frac{\partial^{3} p_{0}}{\partial x^{2}\,\partial \zeta}\,d\zeta
\end{equation*}
\begin{equation*}
\int_{\xi_{E}}^{\xi}\frac{\alpha'}{\alpha^{2}}\,\frac{\partial^{3} p_{0}}{\partial x^{2}\,\partial \zeta}\,d\zeta=
-\big[\frac{1}{\alpha}\,\frac{\partial^{3} p_{0}}{\partial x^{2}\,\partial \zeta}\big]_{\xi_{E}}^{\xi}
-\int_{\xi_{E}}^{\xi}\frac{\psi''}{\alpha}\,\frac{\partial^{4} p_{0}}{\partial x^{3}\,\partial \zeta}\,d\zeta+
\int_{\xi_{E}}^{\xi}\frac{1}{\alpha}\,\frac{\partial^{4} p_{0}}{\partial x^{2}\,\partial \zeta^{2}}\,d\zeta
\end{equation*}
and
\begin{align*}\
\int_{\xi_{E}}^{\xi}\frac{\alpha''}{\alpha^{2}}\,\frac{\partial^{2} p_{0}}{\partial x^{2}}\,d\zeta&=
\big[\frac{\alpha'}{\alpha^{2}}\,\frac{\partial^{2} p_{0}}{\partial x^{2}}\big]_{\xi_{E}}^{\xi}+
2\,\int_{\xi_{E}}^{\xi}\frac{(\alpha')^{2}}{\alpha^{3}}\,\frac{\partial^{2} p_{0}}{\partial x^{2}}\,d\zeta+
\int_{\xi_{E}}^{\xi}\psi''\frac{\alpha'}{\alpha^{2}}\,
\frac{\partial^{3} p_{0}}{\partial x^{3}}\,d\zeta+
\big[\frac{1}{\alpha}\,\frac{\partial^{3} p_{0}}{\partial x^{2}\,\partial \zeta}\big]_{\xi_{E}}^{\xi}\ \\
&+\int_{\xi_{E}}^{\xi}\frac{\psi''}{\alpha}\,\frac{\partial^{4} p_{0}}{\partial x^{3}\,\partial \zeta}\,d\zeta
-\int_{\xi_{E}}^{\xi}\frac{1}{\alpha}\,\frac{\partial^{4} p_{0}}{\partial x^{2}\,\partial \zeta^{2}}\,d\zeta
\end{align*}
It immediately follows that
\begin{align*}
\int_{\xi_{E}}^{\xi}\frac{(\partial_{x} f_{0})\big(-\psi'(\zeta),\zeta\big)}{\alpha(\zeta)}\,d\zeta&=
-3\,\int_{\xi_{E}}^{\xi}\frac{1}{\alpha}\,\frac{\partial^{4} p_{0}}{\partial x^{2}\,\partial \zeta^{2}}\,d\zeta+
2\,\int_{\xi_{E}}^{\xi}\frac{\psi''}{\alpha}\,\frac{\partial^{4} p_{0}}{\partial x^{3}\,\partial \zeta}\,d\zeta+
\int_{\xi_{E}}^{\xi}\frac{(\psi'')^{2}}{\alpha}\,\frac{\partial^{4} p_{0}}{\partial x^{4}}\,d\zeta\ \\
&+3\,\int_{\xi_{E}}^{\xi}\frac{(\alpha')^{2}}{\alpha^{3}}\,\frac{\partial^{2} p_{0}}{\partial x^{2}}\,d\zeta
+4\,\int_{\xi_{E}}^{\xi}\psi''\frac{\alpha'}{\alpha^{2}}\,\frac{\partial^{3} p_{0}}{\partial x^{3}}\,d\zeta+
\big[\frac{\psi''}{\alpha}\,\frac{\partial^{3} p_{0}}{\partial x^{3}}\big]_{\xi_{E}}^{\xi}\ \\
&+\big[\frac{\alpha'}{\alpha^{2}}\,\frac{\partial^{2} p_{0}}{\partial x^{2}}\big]_{\xi_{E}}^{\xi}
+3\,\big[\frac{1}{\alpha}\,\frac{\partial^{3} p_{0}}{\partial x^{2}\,\partial \zeta}\big]_{\xi_{E}}^{\xi}
\end{align*}
hence
\begin{equation}\label{}
24\,J_{2}=\int_{\xi_{E}}^{\xi}\frac{(\partial_{x} f_{0})\big(-\psi'(\zeta),\zeta\big)}{\alpha(\zeta)}\,d\zeta
+3\,\big[\frac{\psi''}{\alpha}\,\frac{\partial^{3} p_{0}}{\partial x^{3}}\big]_{\xi_{E}}^{\xi}+
5\,\big[\frac{\alpha'}{\alpha^{2}}\,\frac{\partial^{2} p_{0}}{\partial x^{2}}\big]_{\xi_{E}}^{\xi}-3\,
\big[\frac{1}{\alpha}\,\frac{\partial^{3} p_{0}}{\partial x^{2}\,\partial \zeta}\big]_{\xi_{E}}^{\xi}
\end{equation}
and modulo the integral of an exact form in $\mathcal{A}$
\begin{align*}
J_{2}&\equiv\frac{1}{24}\,\int_{\xi_{E}}^{\xi}\frac{(\partial_{x} f_{0})\big(-\psi'(\zeta),\zeta\big)}{\alpha(\zeta)}\,d\zeta
     =\frac{1}{24}\,\int_{\gamma_{E}}\frac{\partial_{x} f_{0}(x,\xi)}{\partial_{x} p_{0}(x,\xi)}\,d\xi\ \\
     &=-\frac{1}{24}\,\int_{0}^{T(E)}\partial_{x} f_{0}\big(x(t),\xi(t)\big)\,dt
     =-\frac{1}{24}\,\frac{d}{dE}\,\int_{\gamma_{E}}f_{0}(x,\xi)\,d\xi\ \\
     &=-\frac{1}{24}\,\frac{d}{dE}\,\int_{\gamma_{E}}\frac{\Delta(x,\xi)}{\partial_{x} p_{0}(x,\xi)}\,\,d\xi
     =\frac{1}{24}\,\frac{d}{dE}\,\int_{0}^{T(E)}\Delta\big(x(t),\xi(t)\big)\,dt
\end{align*}
Using these expressions, we recover the well known action integrals (see e.g. \cite{CdV1}):
We know that
\begin{equation*}
\frac{d}{dE}\,\int_{0}^{T(E)}\Gamma\big(x(t),\xi(t)\big)\,dt=2\,\int_{0}^{T(E)}\Delta\big(x(t),\xi(t)\big)\,dt
\end{equation*}
where $\Gamma\,dt$ is the restriction to $\gamma_{E}$ of the 1-form $\omega_{0}$ in $\mathbb{R}^{2}$, defined by
\begin{equation*}
\omega_{0}(x,\xi)=\big(\frac{\partial^{2} p_{0}}{\partial x^{2}}\,
\frac{\partial p_{0}}{\partial \xi}-\frac{\partial^{2} p_{0}}{\partial x\,\partial \xi}\,
\frac{\partial p_{0}}{\partial\,x}\big)\,dx+\big(\frac{\partial^{2} p_{0}}{\partial x\,\partial \xi}\,
\frac{\partial p_{0}}{\partial \xi}-
\frac{\partial^{2} p_{0}}{\partial \xi^{2}}\,
\frac{\partial\,p_{0}}{\partial x}\big)\,d\xi
\end{equation*}
By writing
$$\sqrt{2}\,D_{1}(\xi)=\int_{\xi_E}^{\xi}\Omega_{1}(\zeta)$$
we find that
\begin{equation}\label{}
\begin{aligned}
\mathrm{Im}\oint_{\gamma_{E}}\Omega_{1}&=
\frac{1}{24}\,\frac{d}{dE}\,\int_{\gamma_{E}}\Delta\,dt-
\int_{\gamma_{E}}p_{2}\,dt-\frac{1}{2}\,\frac{d}{dE}
\int_{\gamma_{E}}p_{1}^{2}\,dt\ \\
&=\frac{1}{48}\,\big(\frac{d}{dE}\big)^{2}\,\int_{\gamma_{E}}\Gamma\,dt-\int_{\gamma_{E}}p_{2}\,dt-
\frac{1}{2}\,\frac{d}{dE}
\int_{\gamma_{E}}p_{1}^{2}\,dt
\end{aligned}
\end{equation}
Using relation (\ref{Real}), we conclude that
\begin{equation}\label{}
\mathrm{Re}\oint_{\gamma_{E}}\Omega_{1}=0
\end{equation}
\subsection{Well normalized QM mod $\mathcal{O}(h^{2})$ in the spatial representation}
The next task consists in extanding the solutions away from $a_{E}=(x_{E}, \xi_{E})$ in the spatial representation.

Recall that in the Fourier representation and for $\xi$ near $\xi_{E}$, the microlocal solution $\widehat{u^{a}}$ of the eigenvalue equation $\big(P(-hD_{\xi},\xi;h)-E\big)\,\widehat u^{a}(\xi;h)=0$ is given by:
\begin{equation}\label{}
\widehat u^{a}(\xi;h)=e^{\frac{i}{h}\,\psi(\xi)}\,\big(b_{0}(\xi)+h\,b_{1}(\xi)+\mathcal{O}(h^{2})\big)
\end{equation}
Next, applying the inverse semi-classical Fourier transform to $\widehat{u^{a}}$, we obtain:
\begin{equation}\label{}
u^{a}(x;h)=(2\,\pi\,h)^{-1/2}\,\int
e^{\frac{i}{h}\,\big(x\,\xi+\psi(\xi)\big)}\,\big(b_{0}(\xi)+h\,b_{1}(\xi)+\mathcal{O}(h^{2})\big)\,d\xi
\end{equation}
The phase of the oscillatory integral defining $u^{a}$ has two critical points, \(\xi_{+}(x)>\xi_E\) and \(\xi_{-}(x)<\xi_E\). The critical values of the phase are given by:
\[\varphi_{\pm}(x)=x\,\xi_{\pm}(x)+\psi\big(\xi_{\pm}(x)\big)\]
From the relation $x+\psi'(\xi_{\pm}(x))=0$, it follows that:
\begin{equation*}
\partial_{x}\varphi_{\pm}(x)=\xi_{\pm}(x)
\end{equation*}
Since
\begin{equation*}
\varphi_{\pm}(x_{E})=x_{E}\,\xi_{E}+\psi(\xi_{E})=x_{E}\,\xi_{E}
\end{equation*}
we deduce that:
\begin{equation*}
\varphi_{\pm}(x):=\varphi_{\pm}(x_{E},x)=x_{E}\,\xi_{E}+\int_{x_{E}}^{x}\xi_{\pm}(y)\,dy
\end{equation*}
Because the phase has two critical points on the support of $b(\xi;h)$, the contributions from each critical point must be summed. By the stationary phase theorem (\ref{Phasestat0}), we obtain:
\begin{equation}\label{}
\int e^{\frac{i}{h}\,\big(x\,\xi+\psi(\xi)\big)}\,b(\xi;h)\,d\xi=
\sum_{\pm}e^{\frac{i}{h}\,\varphi_{\pm}(x)}\,\Big(\frac{\psi''\big(\xi_{\pm}(x)\big)}{2\,i\,\pi\,h}\Big)^{-\frac{1}{2}}\,
\Big(b_{0}\big(\xi_{\pm}(x)\big)+h\,b_{1}\big(\xi_{\pm}(x)\big)+h\,L_{1}b_{0}\big(\xi_{\pm}(x)\big)+\mathcal{O}(h^{2})\Big)
\end{equation}
where
\begin{equation}\label{LtildeopeRATOR}
L_{1}b_{0}\big(\xi_{\pm}(x)\big)=\sum_{n=0}^{2}\frac{2^{-(n+1)}}{i\,n!\,(n+1)!}\,\bigg<\Big(\psi''\big(\xi_{\pm}(x)\big)\Big)^{-1}\,D_{\xi},
D_{\xi}\bigg>^{n+1}\,(\phi_{x}^{n}\,b_{0})\big(\xi_{\pm}(x)\big)
\end{equation}
and
\begin{equation}\label{}
\phi_{x}(\xi)=x\,\big(\xi-\xi_{\pm}(x)\big)+\psi(\xi)-\psi(\xi_{\pm}(x))-\frac{1}{2}\,\psi''\big(\xi_{\pm}(x)\big)\big(\xi-\xi_{\pm}(x)\big)^{2}=
\mathcal{O}\Big(\big(\xi-\xi_{\pm}(x)\big)^{3}\Big)
\end{equation}
A straightforward calculation shows that:
$$\big<\big(\psi''(\xi_{\pm}(x))\big)^{-1}\,D_{\xi},
D_{\xi}\big>\,b_{0}(\xi_{\pm}(x))=\bigg[-\big(\psi''(\xi)\big)^{-1}\,b_{0}''(\xi)\bigg]_{\xi=\xi_{\pm}(x)}$$
$$\big<\big(\psi''(\xi_{\pm}(x))\big)^{-1}\,D_{\xi},
D_{\xi}\big>^{2}\,(\phi_{x}\,b_{0})(\xi_{\pm}(x))=\bigg[\big(\psi''(\xi)\big)^{-2}\,\big(\psi^{(4)}(\xi)\,b_{0}(\xi)+
4\,\psi^{(3)}(\xi)\,b'_{0}(\xi)\big)\bigg]_{\xi=\xi_{\pm}(x)}$$
and
$$\big<\big(\psi''(\xi_{\pm}(x))\big)^{-1}\,D_{\xi},
D_{\xi}\big>^{3}\,(\phi_{x}^{2}\,b_{0})(\xi_{\pm}(x))=-20\,\bigg[\big(\psi''(\xi)\big)^{-3}\,\big(\psi'''(\xi)\big)^{2}\,b_{0}(\xi)
\bigg]_{\xi=\xi_{\pm}(x)}$$
In a neighborhood of the focal point $a_{E}$ and for $x<x_{E}$, the microlocal solution of $\big(P(x,hD_{x})-E\big)u(x;h)=0$ is given mod $\mathcal{O}(h^{2})$ by:
\begin{equation}\label{uaxH0UAXH}
\begin{aligned}
\!\!\!u^{a}_{\pm}(x;h)&=\sum_{\pm}u^{a}_{\pm}(x;h)\ \\
                &=2^{-1/2}\,\sum_{\pm}e^{\pm i\,\frac{\pi}{4}}\,\Big(\pm \partial_{\xi}p_{0}\big(x,\xi_{\pm}(x)\big)\Big)^{-\frac{1}{2}}
\,\exp\bigg[\frac{i}{h}\,\Big(\varphi_{\pm}(x)-h\,\int_{x_{E}}^{x}
\frac{p_{1}\big(y,\xi_{\pm}(y)\big)}{\partial_{\xi}p_{0}\big(y,\xi_{\pm}(y)\big)}\,dy\Big)\,
\bigg]\,\Big(1+h\,\frac{b_{1}\big(\xi_{\pm}(x)\big)}{b_{0}\big(\xi_{\pm}(x)\big)}+h\,D_{2}\big(\xi_{\pm}(x)\big)\Big),
\end{aligned}
\end{equation}
with $\pm\partial_{\xi}p_{0}\big(x,\xi_{\pm}(x)\big)>0$, and where we define:
\begin{equation}\label{}
D_{2}(\xi)=-\frac{1}{2i}\,(\psi''(\xi))^{-1}\,\frac{b''_{0}(\xi)}{b_{0}(\xi)}+
\frac{1}{8i}\,(\psi''(\xi))^{-2}\,\bigg(\psi^{(4)}(\xi)+4\,\psi^{(3)}(\xi)
\,\frac{b'_{0}(\xi)}{b_{0}(\xi)}\bigg)-\frac{5}{24i}\,(\psi''(\xi))^{-3}\,
(\psi^{(3)}(\xi))^{2}
\end{equation}
and $\psi^{(j)}, j\geq 3$, denotes the $j-\text{th}$ derivative of $\psi$. It is also easy to see that
\begin{equation}\label{b1surbzero}
\frac{b_{1}(\xi)}{b_{0}(\xi)}=\sqrt{2}\,\big(C_{1}(E)+D_{1}(\xi)\big)=-\frac{1}{2}\,\partial_{x}\big(\frac{p_{1}}{\partial_{x} p_{0}}\big)(-\psi'(\xi),\xi)+i\,\sqrt{2}\,\mathrm{Im}\big(D_{1}(\xi)\big)
\end{equation}
We also have
$$\frac{b'_{0}(\xi)}{b_{0}(\xi)}=-\frac{\alpha'(\xi)}{2\,\alpha(\xi)}+\frac{i\,p_{1}\big(-\psi'(\xi),\xi\big)}{\alpha(\xi)}$$
and
$$\frac{b''_{0}(\xi)}{b_{0}(\xi)}=\Big(-\frac{\alpha'(\xi)}{2\,\alpha(\xi)}+\frac{i\,p_{1}\big(-\psi'(\xi),\xi\big)}{\alpha(\xi)}\Big)^{2}+
\frac{d}{d\xi}\Big(-\frac{\alpha'(\xi)}{2\,\alpha(\xi)}+\frac{i\,p_{1}\big(-\psi'(\xi),\xi\big)}{\alpha(\xi)}\Big)$$
First, we observe that $D_{2}\big(\xi_{\pm}(x)\big)$ does not contribute to the homology class of the semi-classical forms defining the action, as it contains no integral term. Thus, the phase in (\ref{uaxH0UAXH}) can be replaced, modulo $\mathcal{O}(h^{3})$ by
\begin{equation}
\begin{aligned}\label{Phasespatiale}
S_{\pm}(x_{E},x;h)&=x_{E}\,\xi_{E}+\int_{x_{E}}^{x}\xi_{\pm}(y)\,dy-
h\,\int_{x_{E}}^{x}\frac{p_{1}\big(y,\xi_{\pm}(y)\big)}{\partial_{\xi}p_{0}\big(y,\xi_{\pm}(y)\big)}\,dy+
\sqrt{2}\,h^{2}\,\mathrm{Im}\Big(D_{1}\big(\xi_{\pm}(x)\big)\Big) \ \\
                  &=x_{E}\,\xi_{E}+\int_{x_{E}}^{x}\xi_{\pm}(y)\,dy-
h\,\int_{x_{E}}^{x}\frac{p_{1}\big(y,\xi_{\pm}(y)\big)}{\partial_{\xi}p_{0}\big(y,\xi_{\pm}(y)\big)}\,dy+
h^{2}\,\int_{x_{E}}^{x}T_{1}\big(\xi_{\pm}(y)\big)\,\xi'_{\pm}(y)\,dy
\end{aligned}
\end{equation}
%where the residue of $\sqrt{2}\,h^{2}\,\mathrm{Im}\Big(D_{1}\big(\xi_{\pm}(x)\big)\Big)$, modulo the integral of an exact form, is computed as in Lemma \ref{LemmMa10101}.
%First we expand
%$$u^{a}(x)=(2\pi h)^{-1/2}\,\int e^{\frac{i}{h}x\xi}\widehat{u^{a}}(\xi;h)\,d\xi=(2\pi h)^{-1/2}\,\int e^{\frac{i}{h}\,\big(x\xi+\psi(\xi)\big)}\widehat{u^{a}}(\xi;h)\,b(\xi;h)\,d\xi$$
%near $x_{E}$ by stationary phase (\ref{Phasestat0}) mod $\mathcal{O}(h^{2})$, selecting the two critical points $\xi_{\pm}(x)$ near $x_{E}$. The phase functions take the form
%$$\varphi_{\pm}(x)=x\,\xi_{\pm}(x)+\psi\big(\xi_{\pm}(x)\big).$$
%%%%%%%%%%%%%%%%%%%%%%%%%%%%%%%%%%%%%%%%%%%%%%%%%%%%%%%%%%%%%%%%%%%%%%%%%%%%%%%%%%%%%%%%%%%%%%%%%%%%%
%bibliographie
\begin{prop}\label{SpatiAlNormalisation}
In the spatial representation, the microlocal Wronskian near a focal point $a_{E}$ is given by
\begin{equation}\label{0SpatiAmalisation0}
\mathcal{W}^{a}\big(u^{a},\overline{u^{a}}\big)=\mathcal{W}_{+}^{a}\big(u^{a},\overline{u^{a}}\big)-\mathcal{W}^{a}_{-}\big(u^{a},\overline{u^{a}}\big)
=1+\mathcal{O}(h^{2})
\end{equation}
\end{prop}
\begin{proof}
Let $\chi^{a}\in C^{\infty}_{0}(\mathbb{R}^{2})$ be a cut-off function as defined in Subsection \ref{NORMALIsation}. The Weyl symbol of the commutator $\displaystyle \frac{i}{h}\,[P,\chi^{a}]$ is given by:
\begin{equation*}
c(x,\xi;h)=\big(\partial_{\xi}\,p_{0}(x,\xi)+h\,\partial_{\xi}\,p_{1}(x,\xi)\big)\,\chi'_{1}(x)+
\mathcal{O}(h^{2})=c_{0}(x,\xi)+h\,c_{1}(x,\xi)+\mathcal{O}(h^{2})
\end{equation*}
Let:
\begin{equation}\label{FAPM}
B^{a}_{\pm}:=\frac{i}{h}[P,\chi^{a}]_{\pm}u^{a}_{\pm}
\end{equation}
so:
$$B^{a}_{\pm}(x;h)=\frac{1}{2\pi h}\,\int\int e^{\frac{i}{h}\,(x-y)\,\eta}\,c\big(\frac{x+y}{2},\eta;h\big)\,u^{a}_{\pm}(y;h)\,dy\,d\eta$$
For fixed $x$, the phase of the oscillatory integral defining $F^{a}_{\pm}(x;h)$ is:
\begin{equation*}
\phi_{x}^{\pm}(y,\eta)=(x-y)\,\eta+\varphi_{\pm}(y).
\end{equation*}
Its critical points are:
\begin{equation*}
(y_{c}(x),\eta_{c}^{\pm}(x))=\big(x,\xi_{\pm}(x)\big),
\end{equation*}
and the corresponding critical values are:
\begin{equation*}
\phi_{x}^{\pm}\big(y_{c}(x),\eta_{c}^{\pm}(x)\big)=\varphi_{\pm}(x).
\end{equation*}
A direct calculation shows that the Hessian matrix is:
\begin{equation*}
\big(\mathrm{Hess}\,\phi^{\pm}_{x}\big)(x,\xi_{\pm}(x))=
\left(
\begin{array}{cc}
\xi'_{\pm}(x) & -1 \\
-1 & 0 \\
\end{array}
\right)
\end{equation*}
We define:
$$u_{x}^{\pm}(y,\eta;h)=c(\frac{x+y}{2},\eta;h)\,
\big(\pm\,\partial_{\xi}p_{0}(y,\xi_{\pm}(y))\big)^{-\frac{1}{2}}\,\exp\big[-i\,\int_{x_{E}}^{y}
\frac{p_{1}(z,\xi_{\pm}(z))}{\partial_{\xi}p_{0}(z,\xi_{\pm}(z))}\,dz\big]\,\big(1+h\,Z(\xi_{\pm}(y))+\mathcal{O}(h^{2})\big)$$
where:
$$Z(\xi_{\pm}(y))=-\frac{1}{2}\,\partial_{x}\big(\frac{p_{1}}{\partial_{x} p_{0}}\big)(y,\xi_{\pm}(y))+i\,\sqrt{2}\,\mathrm{Im}\big(D_{1}(\xi_{\pm}(y))\big)+
D_{2}(\xi_{\pm}(y))$$
Thus, the leading term of $u_{x}^{\pm}(y,\eta;h)$ is:
$$u_{x}^{(0,\pm)}(y,\eta)=c_{0}(\frac{x+y}{2},\eta)\,
\big(\pm\,\partial_{\xi}p_{0}(y,\xi_{\pm}(y))\big)^{-\frac{1}{2}}\,\exp\big[-i\,\int_{x_{E}}^{y}
\frac{p_{1}(z,\xi_{\pm}(z))}{\partial_{\xi}p_{0}(z,\xi_{\pm}(z))}\,dz\big]=c_{0}(\frac{x+y}{2},\eta)\,v_{\pm}(y).$$
By the stationary phase theorem (\ref{Phasestat0}), we obtain:
\begin{equation*}
B^{a}_{\pm}(x;h)=\frac{1}{\sqrt{2}}\,e^{\pm i\,\frac{\pi}{4}}\,e^{\frac{i}{h}\,
\varphi_{\pm}(x)}\,\bigg(u_{x}^{\pm}\big(x,\xi_{\pm}(x);h\big)+h\,L_{1}u_{x}^{(0,\pm)}\big(x,\xi_{\pm}(x)\big)+\mathcal{O}(h^{2})\bigg)
\end{equation*}
where:
\begin{equation}\label{}
L_{1}u_{x}^{(0,\pm)}\big(x,\xi_{\pm}(x)\big)=\sum_{n=0}^{2}\frac{2^{-(n+1)}}{i\,n!\,(n+1)!}\,\big(2\,\frac{\partial^{2}}{\partial y\,\partial \eta}+\xi'_{\pm}(x)\,\frac{\partial^{2}}{\partial \eta^{2}}\big)^{n+1}(\psi^{n}_{x}\,u^{(0,\pm)}_{x})(x,\xi_{\pm}(x))
\end{equation}
and
\begin{equation}\label{}
\psi_{x}(y,\eta)=(x-y)\,\xi_{\pm}(x)+\varphi_{\pm}(y)-\varphi_{\pm}(x)-\frac{1}{2}\,\xi'_{\pm}(x)\,(y-x)^{2}=\mathcal{O}\big((y-x)^{3}\big)
\end{equation}
A few calculations show that:
$$\big(2\,\frac{\partial^{2}}{\partial y\,\partial \eta}+\xi'_{\pm}(x)\,\frac{\partial^{2}}{\partial \eta^{2}}\big)u^{(0,\pm)}_{x}(x,\xi_{\pm}(x))=2\,v'_{\pm}(x)\,s_{\pm}(x)+v_{\pm}(x)\,s'_{\pm}(x),$$
where:
$$s_{\pm}(x)=(\frac{\partial c_{0}}{\partial \xi})(x,\xi_{\pm}(x)).$$
Moreover:
$$\big(2\,\frac{\partial^{2}}{\partial y\,\partial \eta}+\xi'_{\pm}(x)\,\frac{\partial^{2}}{\partial \eta^{2}}\big)^{n+1}(\psi^{n}_{x}\,u^{(0,\pm)}_{x})(x,\xi_{\pm}(x))=0,\quad\forall n\in \{1,2\}.$$
It is easy to see that:
$$v'_{\pm}(x)=\theta_{\pm}(x)\,v_{\pm}(x),$$
where:
$$\theta_{\pm}(x)=-
\frac{1}{\psi''(\xi_{\pm}(x))\,\alpha(\xi_{\pm}(x))}\,\bigg(i\,p_{1}\big(x,\xi_{\pm}(x)\big)-
\frac{\psi'''(\xi_{\pm}(x))\,\alpha(\xi_{\pm}(x))+\psi''(\xi_{\pm}(x))\,
\alpha'(\xi_{\pm}(x))}{2\,\psi''(\xi_{\pm}(x))}\bigg)$$
and:
$$c_{0}\big(x,\xi_{\pm}(x)\big)\,\big(\pm\,\partial_{\xi}p_{0}(x,\xi_{\pm}(x))\big)^{-\frac{1}{2}}=
\pm\,\big(\pm\,\partial_{\xi}p_{0}(x,\xi_{\pm}(x))\big)^{\frac{1}{2}}\,\chi'_{1}(x)$$
Consequently:
\begin{align*}
B^{a}_{\pm}(x;h)&=\pm\,2^{-1/2}\,e^{\pm i\,\frac{\pi}{4}}\,\exp\bigg[\frac{i}{h}\,\bigg(\varphi_{\pm}(x)-h\,\int_{x_{E}}^{x}
\frac{p_{1}\big(y,\xi_{\pm}(y)\big)}{\partial_{\xi}p_{0}\big(y,\xi_{\pm}(y)\big)}\,dy\bigg)\,
\bigg]\,\big(\pm\,\partial_{\xi}p_{0}(x,\xi_{\pm}(x))\big)^{\frac{1}{2}}\,\chi'_{1}(x)\ \\
&\times \bigg(1+h\,Z(\xi_{\pm}(x))+h\,\frac{c_{1}(x,\xi_{\pm}(x))}{c_{0}(x,\xi_{\pm}(x))}+
\big(\frac{2\,s_{\pm}(x)\,\theta_{\pm}(x)+s'_{\pm}(x)}{2\,i\,c_{0}(x,\xi_{\pm}(x))}\big)\,h+\mathcal{O}(h^{2})\bigg).
\end{align*}
Next, observing that:
$$s_{\pm}(x)=(\frac{\partial^{2} p_{0}}{\partial \xi^{2\,}})(x,\xi_{\pm}(x))\,\chi'_{1}(x)=\omega_{\pm}(x)\,\chi'_{1}(x),$$
and that:
$$\partial_{\xi} p_{0}(x,\xi_{\pm}(x))=\psi''(\xi_{\pm}(x))\,\alpha(\xi_{\pm}(x)),$$
we obtain:
\begin{equation}
\begin{aligned}\label{ASYmPTOTiCFa}
B^{a}_{\pm}(x;h)&=\pm\,2^{-1/2}\,e^{\pm i\,\frac{\pi}{4}}\,\exp\bigg[\frac{i}{h}\,\bigg(\varphi_{\pm}(x)-h\,\int_{x_{E}}^{x}
\frac{p_{1}\big(y,\xi_{\pm}(y)\big)}{\partial_{\xi}p_{0}\big(y,\xi_{\pm}(y)\big)}\,dy\bigg)\,
\bigg]\,\big(\pm\,\partial_{\xi}p_{0}(x,\xi_{\pm}(x))\big)^{\frac{1}{2}}\,\chi'_{1}(x)\ \\
&\times \bigg(1+h\,Z(\xi_{\pm}(x))+h\,\frac{\partial_{\xi} p_{1}(x,\xi_{\pm}(x))}{\partial_{\xi} p_{0}(x,\xi_{\pm}(x))}-
\frac{i\,h\,\omega_{\pm}(x)\,\theta_{\pm}(x)}{\partial_{\xi} p_{0}(x,\xi_{\pm}(x))}-\frac{i\,h}{2}\,\frac{\displaystyle\frac{d}{dx}\big(\omega_{\pm}(x)\,\chi'_{1}(x)\big)}{\partial_{\xi} p_{0}(x,\xi_{\pm}(x))\,\chi'_{1}(x)}+\mathcal{O}(h^{2})\bigg).
\end{aligned}
\end{equation}
This gives:
\begin{align*}
(u^{a}_{+}|B^{a}_{+})&=\frac{1}{2}\,\int_{x_{E}}^{+\infty}\chi'_{1}(x)\,dx+\frac{h}{2}\,\int_{x_{E}}^{+\infty}\bigg(2\,\mathrm{Re}\big(Z(\xi_{+}(x))\big)+
\frac{\partial_{\xi} p_{1}(x,\xi_{+}(x))}{\psi''(\xi_{+}(x))\,\alpha(\xi_{+}(x))}+\frac{i\,\omega_{+}(x)\,
\overline{\theta_{+}(x)}}{\psi''(\xi_{+}(x))\,\alpha(\xi_{+}(x))}\bigg)\,\chi'_{1}(x)\,dx\ \\
&+\frac{i\,h}{4}\,\int_{x_{E}}^{+\infty}\frac{1}{\psi''(\xi_{+}(x))\,\alpha(\xi_{+}(x))}\,\displaystyle\frac{d}{dx}\big(\omega_{+}(x)\,\chi'_{1}(x)\big)\,dx+\mathcal{O}(h^{2})\ \\
&=\frac{1}{2}+\frac{h}{2}\,K_{1}+\frac{i\,h}{4}\,K_{2}+\mathcal{O}(h^{2}).
\end{align*}
A simple calculation shows that:
$$2\,\mathrm{Re}\big(Z(\xi_{+}(x))\big)+
\frac{\partial_{\xi} p_{1}(x,\xi_{+}(x))}{\psi''(\xi_{+}(x))\,\alpha(\xi_{+}(x))}+\frac{i\,\omega_{+}(x)\,
\overline{\theta_{+}(x)}}{\psi''(\xi_{+}(x))\,\alpha(\xi_{+}(x))}=
\frac{\omega_{+}(x)}{\psi''(\xi_{+}(x))\,\alpha(\xi_{+}(x))}
\,\bigg(i\,\overline{\theta_{+}(x)}+\frac{p_{1}(x,\xi_{+}(x))}{\psi''(\xi_{+}(x))\,\alpha(\xi_{+}(x))}\bigg)$$
$$=\frac{i\,\omega_{+}(x)}{2\,\big(\psi''(\xi_{+}(x))\big)^{3}\,\big(\alpha(\xi_{+}(x))\big)^{2}}\,\bigg(\psi'''(\xi_{+}(x))\,\alpha(\xi_{+}(x))+
\psi''(\xi_{+}(x))\,\alpha'(\xi_{+}(x))\bigg).$$
Hence:
$$K_{1}=\frac{i}{2}\,\int_{x_{E}}^{+\infty}\frac{\omega_{+}(x)}{\big(\psi''(\xi_{+}(x))\big)^{3}\,\big(\alpha(\xi_{+}(x))\big)^{2}}\,
\bigg(\psi'''(\xi_{+}(x))\,\alpha(\xi_{+}(x))+
\psi''(\xi_{+}(x))\,\alpha'(\xi_{+}(x))\bigg)\,\chi'_{1}(x)\,dx.$$
Here, we used the fact that:
$$\omega_{+}(x):=(\frac{\partial^{2} p_{0}}{\partial \xi^{2}})(x,\xi_{+}(x))=\psi'''(\xi_{+}(x))\,\alpha(\xi_{+}(x))+2\,\psi''(\xi_{+}(x))\,\alpha'(\xi_{+}(x))+\big(\psi''(\xi_{+}(x))\big)^{2}\,(\frac{\partial^{2} p_{0}}{\partial x^{2}})(x,\xi_{+}(x)).$$
Integrating by parts gives:
\begin{align*}
K_{2}&=\bigg[\frac{\omega_{+}(x)\,\chi'_{1}(x)}{\psi''(\xi_{+}(x))\,\alpha(\xi_{+}(x))}\bigg]_{x_{E}}^{+\infty}-
\int_{x_{E}}^{+\infty}\frac{d}{dx}\big(\frac{1}{\psi''(\xi_{+}(x))\,\alpha(\xi_{+}(x))}\big)\,\omega_{+}(x)\,\chi'_{1}(x)\,dx\ \\
&=-\int_{x_{E}}^{+\infty}\frac{\omega_{+}(x)}{\big(\psi''(\xi_{+}(x))\big)^{3}\,\big(\alpha(\xi_{+}(x))\big)^{2}}\,
\bigg(\psi'''(\xi_{+}(x))\,\alpha(\xi_{+}(x))+
\psi''(\xi_{+}(x))\,\alpha'(\xi_{+}(x))\bigg)\,\chi'_{1}(x)\,dx\ \\
&=2\,i\,K_{1}
\end{align*}
Thus, we have:
$$(u^{a}_{+}|B^{a}_{+})=\frac{1}{2}+\mathcal{O}(h^{2}),$$
and similarly:
$$(u^{a}_{-}|B^{a}_{-})=-\frac{1}{2}+\mathcal{O}(h^{2}),$$
Consequently:
$$(u^{a}|B^{a}_{+}-B^{a}_{-})=1+\mathcal{O}(h^{2}).$$
Note that the mixed terms $(u^{a}_{\pm}|B^{a}_{\mp})$ are $\mathcal{O}(h^{\infty})$ because the phase is non-stationary.
\end{proof}
\subsection{WKB solutions mod $\mathcal{O}(h^{2})$ in the spatial representation}\label{Section2.1}
We begin by constructing the WKB solutions $u^{a}_{\rho}(x;h)=u^{a}_{\pm}(x;h)$ starting from the focal point $a=a_E$. These solutions are uniformly valid with respect to $h$ for $x$ in any interval $I\subset\subset]x'_E,x_E[$. The solutions take the form:
\begin{equation}\label{2.1}
u^{a}_{\rho}(x;h)=a_{\rho}(x;h)\,e^{\frac{i}{h}\,\varphi_{\rho}(x)},
\end{equation}
where $a_{\rho}(x;h)$ is a formal series in $h$, which we shall compute with $h^2$ accuracy
\begin{equation*}
a_{\rho}(x;h)=a_{\rho,0}(x)+h\,a_{\rho,1}(x)+h^{2}\,a_{\rho,2}(x)+\cdots.
\end{equation*}
The phase $\varphi_{\rho}(x)$ is a real smooth function that satisfies the eikonal equation
\begin{equation}\label{2.2}
p_{0}\big(x,\varphi'_{\rho}(x)\big)=E.
\end{equation}
For simplicity we shall omit indices $\rho=\pm$ whenever no confusion may occur. Let $Q(x,hD_x;h)=e^{-\frac{i}{h}\,\varphi(x)}P(x,hD_x)e^{\frac{i}{h}\,\varphi(x)}$, which is an $h$-pseudo-differential operator. Its action on $a(x;h)$ is given by:
\begin{equation*}
(Q-E)a(x;h)=(2\pi h)^{-1}\,\int\int e^{\frac{i}{h}\,(x-y)\,\theta}\,p\big(\frac{x+y}{2},\theta+F(x,y);h\big)\,a(y;h)\,dy\,d\theta,
\end{equation*}
where $F(x,y)=\displaystyle\int_0^1\varphi'\big(x+t(y-x)\big)\,dt$. Applying stationary phase theorem (\ref{A1}) at order 2 (see Appendix), we find modulo $\mathcal{O}(h^{3})$:
\begin{equation}\label{2.3}
\begin{aligned}
&\big(Q(x,hD_x;h)-E\big)a(x;h)=
\Big(p\big(x,\varphi'(x);h\big)-E\Big)\,a(x;h)
+\frac{h}{i}\,\Big(\beta(x;h)\,\partial_{x}a(x;h)+\frac{1}{2}\,\partial_{x}\beta(x;h)\,a(x;h)\Big)\cr
&-h^{2}\,\Big(\frac{1}{8}\,\partial_{x}r(x;h)\,a(x;h)+\frac{1}{8}\,\varphi''(x)\,\partial_{x}\theta(x;h)\,a(x;h)+\frac{1}{2}\,
\partial_{x}\gamma(x;h)\,\partial_{x}a(x;h)+\frac{1}{2}\,\gamma(x;h)\,\frac{\partial^{2} a(x;h)}{\partial x^{2}}+\frac{1}{6}\,\varphi'''(x)\,
\theta(x;h)\,a(x;h)\Big).\cr
\end{aligned}
\end{equation}
Suppose now that $p(x,\xi;h)$ is real, $p_{0}(x_{E},\xi_{E})=E$, $(\displaystyle\frac{\partial p_{0}}{\partial \xi})(x_{E},\xi_{E})\neq 0$. We look for formal solutions \big(i.e in the sense of formal classical symbols\big) of
\begin{equation}\label{2.31}
\big(P(x,h\,D_{x};h)-E\big)\big(a(x;h)
\,e^{\frac{i}{h}\,\varphi(x)}\big)=0\Leftrightarrow\big(Q(x,h\,D_{x};h)-E\big)a(x;h)=0.
\end{equation}
Once the eikonal equation (\ref{2.2}) is satisfied, the first transport equation is obtained by setting the $\mathcal{O}(h)$ term in (\ref{2.3}) to zero:
\begin{equation}\label{2.4}
\beta_{0}(x)\,a'_{0}(x)+\Big(i\,p_{1}\big(x,\varphi'(x)\big)+\frac{1}{2}\,\beta_{0}'(x)\Big)\,a_{0}(x)=0.
\end{equation}
Its solutions are of the form:
\begin{equation}\label{2.5}
a_{0}(x)=\widetilde{C_{0}}\,|\beta_{0}(x)|^{-\frac{1}{2}}\,\exp\Big(-i\,\int_{x_{E}}^{x}\frac{p_{1}\big(y,\varphi'(y)\big)}{\beta_{0}(y)}\,dy\Big),
\end{equation}
$\widetilde{C_{0}}$ being so far an arbitrary constant.

Next, setting the $\mathcal{O}(h^{2})$ term in (\ref{2.3}) to zero yields a differential equation for $a_{1}(x)$:
\begin{equation}\label{2.6}
\begin{aligned}
&\beta_{0}(x)\,a'_{1}(x)+\Big(i\,p_{1}\big(x,\varphi'(x)\big)+\frac{1}{2}\,\beta_{0}'(x)\Big)
\,a_{1}(x)=-\beta_{1}(x)\,a'_{0}(x)-\Big(i\,p_{2}\big(x,\varphi'(x)\big)+\frac{1}{2}\,\beta'_{1}(x)\Big)\,a_{0}(x)\cr
&+i\,\Big(\frac{1}{8}\,r'_{0}(x)\,a_{0}(x)+\frac{1}{8}\,\varphi''(x)\,\theta'_{0}(x)\,a_{0}(x)+\frac{1}{2}\,
\gamma'_{0}(x)\,a'_{0}(x)+\frac{1}{2}\,\gamma_{0}(x)\,a''_{0}(x)+\frac{1}{6}\,\varphi'''(x)\,\theta_{0}(x)\,a_{0}(x)\Big).
\end{aligned}
\end{equation}
Here, we have introduced the notations:
\begin{equation*}
\beta_{0}(x)=(\frac{\partial p_{0}}{\partial \xi})\big(x,\varphi'(x)\big),\quad r_{0}(x)=
(\frac{\partial^{3} p_{0}}{\partial x \partial \xi^{2}})\big(x,\varphi'(x)\big),\quad\gamma_{0}(x)=
(\frac{\partial^{2} p_{0}}{\partial \xi^{2}})\big(x,\varphi'(x)\big),\quad\theta_{0}(x)=(\frac{\partial^{3} p_{0}}{\partial \xi^{3}})\big(x,\varphi'(x)\big).
\end{equation*}
The homogeneous equation associated with (\ref{2.6}) is the same as (\ref{2.4}); so we are looking for a particular solution of (\ref{2.6}), integrating from $x_E$, of the form
\begin{equation}\label{2.61}
a_1(x)=\widetilde{D_{1}}(x)\,|\beta_{0}(x)|^{-\frac{1}{2}}\,\exp\Big(-i\,\int_{x_{E}}^{x}\frac{p_{1}\big(y,\varphi'(y)\big)}{\beta_{0}(y)}\,dy\Big).
\end{equation}
Alternatively, we could integrate (\ref{2.6}) from $x'_E$ instead of $x_E$. So our main task will consist in computing $\widetilde{D_{1}}(x)$ as a multivalued function, due to the presence of the turning points, in the same way we have determined $D_1(\xi)$ in \cite{IfaLouRo} (Formula (3.5)), using Fourier representation. We solve (\ref{2.6}) by the method of variation of constants, and find
\begin{equation}\label{2.7}
(\widetilde{C_{0}})^{-1}\,\mathrm{Re}\big(\widetilde{D_{1}}(x)\big)=-\frac{1}{2}\,\Big[\partial_{\xi}(\frac{p_{1}}{\partial_{\xi} p_{0}})\big(y,\varphi'(y)\big)\Big]_{x_{E}}^{x},
\end{equation}
\begin{equation}\label{2.8}
\begin{aligned}
&(\widetilde{C_{0}})^{-1}\,\mathrm{Im}\big(\widetilde{D_{1}}(x)\big)=\int_{x_{E}}^{x} \frac{1}{\beta_{0}}\,\bigg(-p_{2}+\frac{1}{8}\,
\frac{\partial^{4} p_{0}}{\partial y^{2} \partial \xi^{2}}+\frac{\varphi''}{12}\,\frac{\partial^{4} p_{0}}{\partial y \partial \xi^{3}}-
\frac{(\varphi'')^{2}}{24}\,\frac{\partial^{4} p_{0}}{\partial \xi^{4}}\bigg)\,dy-
\frac{1}{8}\,\int_{x_{E}}^{x} \frac{(\beta'_{0})^{2}}{\beta_{0}^{3}}\,\frac{\partial^{2} p_{0}}{\partial \xi^{2}}\,dy\cr
&+\frac{1}{6}\,\int_{x_{E}}^{x}\,\varphi''\,\frac{\beta'_{0}}{\beta_{0}^{2}}\,\frac{\partial^{3} p_{0}}{\partial \xi^{3}}\,dy+
\int_{x_{E}}^{x} \frac{p_{1}}{\beta_{0}^{2}}\,\bigg(\partial_{\xi} p_{1}-\frac{p_{1}}{2\,\beta_{0}}\,\frac{\partial^{2} p_{0}}{\partial \xi^{2}}\bigg)\,dy+
\Big[\frac{\varphi''}{6\,\beta_{0}}\,\frac{\partial^{3} p_{0}}{\partial \xi^{3}}-\frac{\beta_{0}'}{4\,\beta_{0}^{2}}\,\frac{\partial^{2} p_{0}}{\partial \xi^{2}}\Big]_{x_{E}}^{x}.
\end{aligned}
\end{equation}
Function $\widetilde{D_{1}}(x)$ can be normalized by
$$\widetilde{D_{1}}(x_{E})=0$$
The general solution of (\ref{2.6}) is then:
\begin{equation}\label{2.9}
a_{1}(x)=\big(\widetilde{C_{1}}+\widetilde{D_{1}}(x)\big)\,|\beta_{0}(x)|^{-\frac{1}{2}}\,
\exp\Big(-i\,\int_{x_{E}}^{x}\frac{p_{1}\big(y,\varphi'(y)\big)}{\beta_{0}(y)}\,dy\Big).
\end{equation}
Consequently,
$$a(x;h)=\Big(\widetilde{C_{0}}+h\,\big(\widetilde{C_{1}}+\widetilde{D_{1}}(x)\big)+\mathcal{O}(h^{2})\Big)\,|\beta_{0}(x)|^{-\frac{1}{2}}\,
\exp\Big(-i\,\int_{x_{E}}^{x}\frac{p_{1}\big(y,\varphi'(y)\big)}{\beta_{0}(y)}\,dy\Big).$$
Repeating this construction for the other branch $(\rho=-1)$ yields the two branches of WKB solutions:
\begin{equation}\label{2.10}
u^{a}_{\pm}(x;h)=|\beta^{\pm}_{0}(x)|^{-\frac{1}{2}}\,e^{\frac{i}{h}\,S_{\pm}(x_{E},x;h)}\,
\Big(\widetilde{C_{0}}+h\,\big(\widetilde{C_{1}}+\widetilde{D_{1}}^{\pm}(x)\big)+\mathcal{O}(h^{2})\Big),
\end{equation}
where
\begin{equation}\label{2.11}
S_{\pm}(x_{E},x;h)=\varphi_{\pm}(x_{E})+\int_{x_{E}}^{x}\xi_{\pm}(y)\,dy-h\,\int_{x_{E}}^{x}\frac{p_{1}\big(y,\varphi'_{\pm}(y)\big)}{\beta_{0}^{\pm}(y)}dy,
\end{equation}
$$\beta^{\pm}_{0}(x)=(\partial_{\xi} p_{0})\big(x,\varphi'_{\pm}(x)\big).$$
Here we have used that $\varphi_{\pm}(x)=\varphi_{\pm}(x_E)+\displaystyle\int_{x_E}^x\xi_{\pm}(y)\,dy$, with $p_0\big(x,\xi_{\pm}(x)\big)=E$.

Normalization with respect to the "flux norm" consists as above in computing $B^{a}_{\pm}=\displaystyle\frac{i}{h}\,[P,\chi^{a}]_{\pm}u^{a}$ by stationary phase theorem (\ref{Phasestat0}) modulo $\mathcal{O}(h^{2})$. Assuming already $\widetilde{C_{0}}$, $\widetilde{C_{1}}$ to be real, a simple calculation using integration by parts
yields $\widetilde{C_{0}}=C_{0}=2^{-1/2}$, and
$$\widetilde{C_{1}}=\widetilde{C_{1}}(a_{E})=-2^{-3/2}\,\partial_{\xi}\big(\frac{p_{1}}{\partial_{\xi}p_{0}}\big)(a_{E}).$$
As a result, outside any neighborhood of $x_{E}$, we have
\begin{equation}\label{2.12}
u_{\pm}(x;h)=|\beta_{0}^{\pm}(x)|^{-1/2}\,e^{\frac{i}{h}\,S_{\pm}(x_{E},x;h)}\,
\big(\widetilde{C_{0}}+h\,\widetilde{C_{1}}+h\,\widetilde{D_{1}}^{\pm}(x)+\mathcal{O}(h^{2})\big),
\end{equation}
with
$$S_{\pm}(x_{E},x;h)=
\varphi_{\pm}(x_{E})+\int_{x_{E}}^{x}\xi_{\pm}(y)\,dy-h\,\int_{x_{E}}^{x}\frac{p_{1}\big(y,\varphi'_{\pm}(y)\big)}{\beta_{0}^{\pm}(y)}\,dy$$
From (\ref{2.12}) we can recover the homology class of generalized action, considering the superposition $u(x;h)=e^{i\pi/4}\,u_{+}(x;h)+e^{-i\pi/4}\,u_{-}(x;h)$ near $a_{E}$. The argument is then similar to that of \cite{AIfaRrou}, formula (1).
\subsection{Bohr-Sommerfeld quantization rule}
Recall from (\ref{Phasespatiale}) the modified phase function of the microlocal solutions $u^{a}_{\pm}$ near the focal point $a_{E}$, accurate to $\mathcal{O}(h^{2})$. Similarly, the phase for the asymptotic solution near the other focal point $a'_{E}$ is given by:
\begin{equation}\label{2.13}
S_{\pm}(x'_{E},x;h)=x'_{E}\,\xi'_{E}+\int_{x'_{E}}^{x}\xi_{\pm}(y)\,dy-
h\,\int_{x'_{E}}^{x}\frac{p_{1}\big(y,\xi_{\pm}(y)\big)}{\partial_{\xi}p_{0}\big(y,\xi_{\pm}(y)\big)}\,dy+
h^{2}\,\int_{x'_{E}}^{x}T_{1}\big(\xi_{\pm}(y)\big)\,\xi'_{\pm}(y)\,dy
\end{equation}
Now, consider the function $B^{a}_{\pm}(x;h)$ with asymptotics (\ref{ASYmPTOTiCFa}), and similarly $B^{a'}_{\pm}(x;h)$. The normalized microlocal solutions
$u^{a}$ and $u^{a'}$, extended uniquely along $\gamma_{E}$ are denoted $u_{1}$ and $u_{2}$. It is then easy to show that
\begin{equation}\label{2.14}
\begin{aligned}
(u_{1}|B^{a'}_{+}-B^{a'}_{-})&=\frac{i}{2}\,\big(e^{\frac{i}{h}\,A_{-}(x_{E},x'_{E};h)}-e^{\frac{i}{h}\,A_{+}(x_{E},x'_{E};h)}\big)\ \\
(u_{2}|B^{a}_{+}-B^{a}_{-})&=\frac{i}{2}\,\big(e^{-\frac{i}{h}\,A_{-}(x_{E},x'_{E};h)}-e^{-\frac{i}{h}\,A_{+}(x_{E},x'_{E};h)}\big)
\end{aligned}
\end{equation}
modulo $\mathcal{O}(h^{2})$. Here, the generalized actions are:
\begin{align*}
A_{\pm}(x_{E},x'_{E};h)&:=S_{\pm}(x_{E},x;h)-S_{\pm}(x'_{E},x;h)\ \\
&=x_{E}\,\xi_{E}-x'_{E}\,\xi'_{E}+\int_{x_{E}}^{x'_{E}}\xi_{\pm}(y)\,dy-h\,\int_{x_{E}}^{x'_{E}}
\frac{p_{1}\big(y,\xi_{\pm}(y)\big)}{\partial_{\xi}p_{0}\big(y,\xi_{\pm}(y)\big)}\,dy+
h^{2}\,\int_{x_{E}}^{x'_{E}}T_{1}\big(\xi_{\pm}(y)\big)\,\xi'_{\pm}(y)\,dy
\end{align*}
The Gram matrix $G^{(a,a')}(E)$ of the solutions $u_1, u_2$ in the basis $\big(B^a_+-B^a_-, B^{a'}_+-B^{a'}_-\big)$ is given by:
\begin{equation}\label{MaTrIcEgRaMm}
G^{(a,a')}(E)=\left(
       \begin{array}{cc}
         1 & \frac{i}{2}\,\big(e^{-\frac{i}{h}\,A_{-}(x_{E},x'_{E};h)}-e^{-\frac{i}{h}\,A_{+}(x_{E},x'_{E};h)}\big)\\
         \frac{i}{2}\,\big(e^{\frac{i}{h}\,A_{-}(x_{E},x'_{E};h)}-e^{\frac{i}{h}\,A_{+}(x_{E},x'_{E};h)}\big) & -1\\
       \end{array}
     \right)
\end{equation}
whose determinant is:
$$-\cos^{2}\big(\frac{A_{-}(x_{E},x'_{E};h)-A_{+}(x_{E},x'_{E};h)}{2h}\big)$$
This determinant vanishes precisely at the eigenvalues of $P$ in $I$ , leading to the condition modulo $\mathcal{O}(h^{3})$:
\begin{align}
2\,\pi\,n\,h&=\int_{x'_{E}}^{x_{E}}\big(\xi_{+}(y)-\xi_{-}(y)\big)\,dy-\pi\,h
-h\,
\int_{x'_{E}}^{x_{E}}\Big(\frac{p_{1}\big(y,\xi_{+}(y)\big)}{\partial_{\xi}p_{0}
\big(y,\xi_{+}(y)\big)}-
\frac{p_{1}\big(y,\xi_{-}(y)\big)}{\partial_{\xi}p_{0}\big(y,\xi_{-}(y)\big)}\Big)\,
dy\nonumber\ \\
&+h^{2}\,\int_{x'_{E}}^{x_{E}}\Big(T_{1}\big(\xi_{+}(y)\big)\,\xi'_{+}(y)-
T_{1}\big(\xi_{-}(y)\big)\,\xi'_{-}(y)\Big)dy,\;\; n\in\mathbb{Z}\label{RCBSG}
\end{align}
We now aim to simplify (\ref{RCBSG}). First, observe that:
$$\int_{x'_{E}}^{x_{E}}\big(\xi_{+}(y)-\xi_{-}(y)\big)\,dy=\oint_{\gamma_{E}}\xi(y)\,dy$$
From the second Hamilton-Jacobi equation:
$$dy(t)=\partial_{\xi}p_{0}\big(y(t),\xi(t)\big)\,dt$$
we derive:
$$\int_{x'_{E}}^{x_{E}}\Big(\frac{p_{1}\big(y,\xi_{+}(y)\big)}{\partial_{y}p_{0}\big(y,\xi_{+}(y)\big)}-
\frac{p_{1}\big(y,\xi_{-}(y)\big)}{\partial_{y}p_{0}\big(y,\xi_{-}(y)\big)}\Big)\,dy=\int_{\gamma_{E}}p_{1}\,dt$$
\begin{align*}
\int_{x'_{E}}^{x_{E}}\Big(T_{1}\big(\xi_{+}(y)\big)\,\xi'_{+}(y)-
T_{1}\big(\xi_{-}(y)\big)\,\xi'_{-}(y)\Big)\,dy&=\oint_{\gamma_{E}}T_{1}\big(\xi(y)\big)\,\xi'(y)\,dy\ \\&=\mathrm{Im}\oint_{\gamma_{E}}\Omega_{_{1}}\big(\xi(y)\big)\ \\
&=\frac{1}{24}\,\frac{d}{dE}\,\int_{\gamma_{E}}\Delta\,dt-
\int_{\gamma_{E}}p_{2}\,dt-\frac{1}{2}\,\frac{d}{dE}
\int_{\gamma_{E}}p_{1}^{2}\,dt
\end{align*}
In conclusion, the generalized Bohr-Sommerfeld quantization rule at second order for an $h$-pseudo-differential operator of the form (\ref{A01WEyLOO}) is given by:
\begin{equation}\label{REGLESEMIGENERALI}
\oint_{\gamma_{E}}\xi(x)\,dx+\Big(-\pi-\int_{\gamma_{E}}p_{1}\,dt\Big)\,h
+h^{2}\,\Big(\frac{1}{24}\,\frac{d}{dE}\,\int_{\gamma_{E}}\Delta\,dt-
\int_{\gamma_{E}}p_{2}\,dt-\frac{1}{2}\,\frac{d}{dE}
\int_{\gamma_{E}}p_{1}^{2}\,dt\Big)=2\,\pi\,n\,h+\mathcal{O}(h^{3})
\end{equation}
\section{Bohr-Sommerfeld and Action-Angle Variables}\label{ACTIO0NANGLE}
Here, we present a simpler approach based on Birkhoff normal form and the monodromy operator, %(\cite{HL1MR1}, \cite{HL2MR2})
reminiscent of \cite{HR84}. Let $P$ be self-adjoint as in (\ref{A01WEyLOO}), with Weyl symbol $p\in S^{0}(m)$, and such that there exists a topological ring $\mathcal{A}$ where $p_{0}$ satisfies the hypotheses $(H_{1})$, $(H_{2})$ and $(H_{3})$ in the Introduction. Without loss of generality, we can assume that $p_{0}$ has a periodic orbit $\gamma_{0}\subset \mathcal{A}$ with period $2\pi$ and energy $E=E_{0}$. From Hamilton-Jacobi theory \big(see \cite{Argyres}\big), there exists a smooth canonical transformation $(t,\tau)\mapsto \kappa (t,\tau)=(x,\xi)$, $t\in [0,2 \pi]$, defined in a neighborhood of $\gamma_{0}$, and a smooth function $\tau\mapsto f_{0}(\tau)$, with $f_{0}(0)=0$ and $f'_{0}(0)=1$, such that
\begin{equation}\label{FNBK0}
p_{0}\circ \kappa(t,\tau)=f_{0}(\tau)
\end{equation}
This transformation is given by generating function $S(\tau,x)=\displaystyle\int_{x_{E}}^{x}\xi(y)\,dy$, where $\xi(x)=\partial_{x}S(\tau,x)$, $\varphi=\partial_{\tau}S(\tau,x)$, and:
$$p_{0}\big(x,\partial_{x}S(\tau,x)\big)=f_{0}(\tau)$$
Energy $E$ and momentum $\tau$ are related by the one-to-one transformation $E=f_{0}(\tau)$, with $f'_{0}(E_{0})=1$.

This map can be quantized semi-classically, known as the semi-classical Birkhoff normal form (BNF). Here, we take advantage of the fact \big(see \cite{CdV1}, Proposition 2\big) that we can smoothly deform $p$ in the interior of annulus $\mathcal{A}$, without changing its semi-classical spectrum in $I$, such that the "new" $p_{0}$ has a non-degenerate minimum at $(x_{0},\xi_{0})$, while all energies $E\in ]0,E_{+}]$ remain regular. The BNF is achieved by introducing "harmonic oscillator" coordinates $(y,\eta)$, so that (\ref{FNBK0}) becomes:
\begin{equation}\label{FNBK0HARMon}
p_{0}\circ \kappa(y,\eta)=f_{0}\big(\frac{1}{2}\,(\eta^{2}+y^{2})\big),
\end{equation}
and $U^{*} P U=f\Big(\frac{1}{2}\,\big((hD_{y})^{2}+y^{2}\big);h\Big)$, has full symbol $f(\tau;h)=f_{0}(\tau)+h\,f_{1}(\tau)+h^{2}\,f_{2}(\tau)+\cdots$. Here, $f_{1}$ includes the Maslov correction 1/2, and $U$ is a microlocally unitary $h$-FIO operator associated with $\kappa$ \big(see \cite{quadzero11}, \cite{HARPE0RIII}\big). In $\mathcal{A}$, $\tau\neq0$, so we can make the smooth symplectic change of coordinates $y=\sqrt{2\tau}\,\cos(t)$, $\eta=\sqrt{2\tau}\,\sin(t)$, and transform $\frac{1}{2}\,\big((hD_{y})^{2}+y^{2}\big)$ back to $hD_{t}$.

We do not provide explicit expressions for $f_{j}(\tau), j\geq 1$, in terms of $p_{j}$, but note that $f_{j}$ depends linearly on $p_{0}, p_{1}, \cdots, p_{j}$ and their derivatives. The BNF effectively eliminates focal points. The Poincaré section $t=0$ in $\{f_{0}(\tau)=E\}=f_{0}^{-1}(E)$ reduces to a single point $\Sigma=\{a(E)\}$.

From \big(\cite{HL1MR1}, \cite{HL2MR2}\big), the Poisson operator $\mathcal{K}(t,E)$ solves (globally near $\gamma_{0}$):
\begin{equation}\label{weyL1}
\big(f(hD_{t};h)-E\big)\,\mathcal{K}(t,E)=0,
\end{equation}
and is given in the special 1-D case by the multiplication operator on $L^{2}(\Sigma)\approx \mathbb{C}^{2}$:
\begin{equation}\label{weyL3}
\mathcal{K}(t,E)=e^{\frac{i}{h}\,S(t,E)}\,a(t,E;h),
\end{equation}
where $S(t,E)$ satisfies the eikonal equation $f_{0}\big(\partial_{t}S(t,E)\big)=E, \,S(0,E)=0$, i.e. $S(t,E)=f_{0}^{-1}(E)\,t=\tau\,t$.
Using formula (\ref{2.3}), we have:
\begin{equation}\label{OSCI0}
\begin{aligned}
&e^{-\frac{i}{h}\,S(t,E)}\,\big(f(hD_{t};h)-E\big)\big(a(t,E;h)\,e^{\frac{i}{h}\,S(t,E)}\big)=
\big(f(\partial_{t} S;h)-E\big)\,a
+\frac{h}{i}\,\Big(\partial_{\tau}f\,\partial_{t}a+\frac{1}{2}\,\frac{\partial^{2} S}{\partial t^{2}}\,\frac{\partial^{2} f}{\partial \tau^{2}}\,a\Big)\cr
&-h^{2}\,\Big(\frac{1}{8}\,\big(\frac{\partial^{2} S}{\partial t^{2}}\big)^{2}\,\frac{\partial^{4} f}{\partial \tau^{4}}\,a+
\frac{1}{2}\,\frac{\partial^{2} S}{\partial t^{2}}\,\frac{\partial^{3} f}{\partial \tau^{3}}\,\partial_{t} a+\frac{1}{2}\,
\frac{\partial^{2} f}{\partial \tau^{2}}\,\frac{\partial^{2} a}{\partial t^{2}}+\frac{1}{6}\,\frac{\partial^{3} S}
{\partial t^{3}}\,\frac{\partial^{3} f}{\partial \tau^{3}}\,a\Big)+\mathcal{O}(h^{3})
\end{aligned}
\end{equation}
where we have simplified the notation by setting:
$$S=S(t,E),\quad a=a(t,E;h),\quad\partial_{\tau}f=(\partial_{\tau}f)\big(\partial_{t} S(t,E);h\big),\quad
\frac{\partial^{2} f}{\partial \tau^{2}}=
\frac{\partial^{2} f}{\partial \tau^{2}}\big(\partial_{t} S(t,E);h\big),\quad
\frac{\partial^{3} f}{\partial \tau^{3}}=\frac{\partial^{3} f}{\partial \tau^{3}}
\big(\partial_{t} S(t,E);h\big),\ldots$$
From relation
$$\frac{\partial^{2} S(t,E)}{\partial t^{2}}=0=\frac{\partial^{3} S(t,E)}{\partial t^{3}},$$
it follows that
\begin{equation}
\begin{aligned}\label{OSCI01}
e^{-\frac{i}{h}\,S(t,E)}\,\big(f(hD_{t};h)-E\big)\big(a(t,E;h)\,e^{\frac{i}{h}\,S(t,E)}\big)&=
\big(f(\tau;h)-E\big)\,a(t,E;h)+\frac{h}{i}\,\partial_{\tau}f(\tau;h)\,\partial_{t}a(t,E;h)\ \\
&-\frac{h^{2}}{2}\,\frac{\partial^{2} f(\tau;h)}{\partial \tau^{2}}\,\frac{\partial^{2} a(t,E;h)}{\partial t^{2}}+\mathcal{O}(h^{3})
\end{aligned}
\end{equation}
If the eikonal equation is satisfied, we obtain by eliminating the $h$ term in (\ref{OSCI01}) the first transport equation:
\begin{equation}\label{TRANSPoRT1}
f'_{0}(\tau)\,\partial_{t} a_{0}(t,E)+i\,f_{1}(\tau)\,a_{0}(t,E)=0
\end{equation}
whose solutions are the functions:
\begin{equation}\label{Sol0}
a_{0}(t,E)=C_{0}\,e^{-i\,t\,f_{1}(\tau)/f'_{0}(\tau)},\quad C_{0}\in \mathbb{R}
\end{equation}
By eliminating the $h^{2}$ term in (\ref{OSCI01}), we obtain the second transport equation:
\begin{equation}\label{TRANSPoRT12}
f'_{0}(\tau)\,\partial_{t} a_{1}(t,E)+i\,f_{1}(\tau)\,a_{1}(t,E)=
-\big(f'_{1}(\tau)\,\partial_{t} a_{0}(t,E)+i\,f_{2}(\tau)\,
a_{0}(t,E)\big)+\frac{i}{2}\,f''_{0}(\tau)\,\frac{\partial^{2} a_{0}(t,E)}{\partial t^{2}}
\end{equation}
The homogeneous equation associated with (\ref{TRANSPoRT12}) is the same as (\ref{TRANSPoRT1}), whose solutions are the functions:
$$t\mapsto C_{1}\,e^{-i\,t\,f_{1}(\tau)/f'_{0}(\tau)},\quad C_{1}\in \mathbb{R}$$
Thus, we seek a particular solution to (\ref{OSCI01}) of the form:
$$t\mapsto D_{1}(t,E)\,e^{-i\,f_{1}(\tau)/f'_{0}(\tau)}$$
Using the method of variation of constants, we find:
\begin{equation}\label{D1tE}
D_{1}(t,E)=i\,C_{0}\,\widetilde{S_{2}}(E)\,t;\quad \widetilde{S_{2}}(E)=\frac{1}{f'_{0}(\tau)}\,\Big(\frac{1}{2}\,\big(\frac{f_{1}^{2}}{f'_{0}}\big)'(\tau)-f_{2}(\tau)\Big)
\end{equation}
which we normalize by setting $D_{1}(0,E)=0$, so that the general solution to (\ref{TRANSPoRT12}) is:
$$a_{1}(t,E)=\big(C_{1}+D_{1}(t,E)\big)\,e^{-i\,t\,f_{1}(\tau)/f'_{0}(\tau)}$$
It follows that:
$$a(t,E;h)=\big(C_{0}+h\,C_{1}+i\,h\,C_{0}\,\widetilde{S_{2}}(E)\,t+\mathcal{O}(h^{2})\big)\,e^{i\,t\,\widetilde{S_{1}}(E)};\quad \widetilde{S_{1}}(E)=-\big(\frac{f_{1}}{f'_{0}}\big)(\tau)$$
and thus:
\begin{equation}\label{KTE}
\mathcal{K}(t,E)=e^{\frac{i}{h}\,\big(S(t,E)+h\,t\,\widetilde{S_{1}}(E)\big)}\,\big(C_{0}+h\,C_{1}+i\,h\,C_{0}\,\widetilde{S_{2}}(E)\,t+\mathcal{O}(h^{2})\big)
\end{equation}
We define the adjoint $\mathcal{K}^{*}(t,E)$ of $\mathcal{K}(t,E)$ by:
$$\mathcal{K}^{*}(t,E)=a^{*}(t,E;h)\,e^{-\frac{i}{h}\,S(t,E)}=\overline{a(t,E;h)}\,e^{-\frac{i}{h}\,S(t,E)}$$ et $$\mathcal{K}^{*}(E)\mathbf{\cdot}=\int\mathcal{K}^{*}(t,E)\mathbf{\cdot}\,dt$$
The "flux norme" on $\mathbb{C}^{2}$, is defined by:
\begin{equation}\label{}
(u|v)_{\chi}=\big(\frac{i}{h}\,[f(hD_{t};h), \chi(t)]\,\mathcal{K}(t,E)u|\mathcal{K}(t,E)v\big)
\end{equation}
with the scalar product of $L^{2}(\mathbb{R}_{t})$ on the RHS, and $\chi \in C^{\infty}(\mathbb{R})$ is a smooth step-function, equal to 0 for $t\leq 0$ an 1 for $t\geq 2\pi$. To normalize the solution $\mathcal{K}(t,E)$, we start from:
$$\mathcal{K}^{*}(E)\,\frac{i}{h}\,[f(hD_{t};h), \chi(t)]\,\mathcal{K}(t,E)=\mathrm{Id}_{L^{2}(\mathbb{R})}$$
The Weyl symbol of $\displaystyle \frac{i}{h}\,[f(hD_{t};h), \chi(t)]$ is $$c(t,\tilde{\tau};h)=\{f(\tilde{\tau};h), \chi(t) \}=\big(f'_{0}(\tilde{\tau})+h\,f'_{1}(\tilde{\tau})\big)\,\chi'(t)+\mathcal{O}(h^{2})=Q(\tilde{\tau};h)\,\chi'(t)+\mathcal{O}(h^{2})
\quad\text{(with $hD_{t}$ quantizing $\tilde{\tau}$)}$$
where we have set:
$$Q(\tilde{\tau};h)=f'_{0}(\tau)+h\,f'_{1}(\tau)$$
We set:
\begin{equation}\label{}
I(t,E):=\frac{i}{h}\,[f(hD_{t};h),\chi(t)]\,\mathcal{K}(t,E)=(2\,\pi\,h)^{-1}\,\int\int e^{\frac{i}{h}\,\big((t-s)\,\tilde{\tau}+S(s,E)\big)}\,Q(\tilde{\tau};h)\,\chi'(\frac{t+s}{2})\,a(s,E;h)\,ds\,d\tilde{\tau}
\end{equation}
For fixed $t$, the phase corresponding to the oscillatory integral defining $I(t,E)$ is:
\begin{equation*}
\varphi_{t}(s,\tilde{\tau})=(t-s)\,\tilde{\tau}+S(s,E)
\end{equation*}
whose critical points are:
\begin{equation*}
\big(s_{c}(t),\tilde{\tau}_{c}(t)\big)=\big(t,\partial_{t} S(t,E)\big)=(t,\tau)
\end{equation*}
and the corresponding critical values are:
\begin{equation*}
\varphi_{t}\big(s_{c}(t),\tilde{\tau}_{c}(t)\big)=S(t,E)
\end{equation*}
A direct calculation shows that:
\begin{equation*}
\mathrm{Hess} \varphi_{t}(s,\tilde{\tau})=
\left(
\begin{array}{cc}
0 & -1 \\
-1 & 0 \\
\end{array}
\right)
\end{equation*}
We set $u_{t}(s,\tilde{\tau};h)=Q(\tilde{\tau};h)\,\chi'(\displaystyle \frac{t+s}{2})\,a(s,E;h)$. By the stationary phase theorem (\ref{Phasestat0}), we obtain:
\begin{equation}\label{}
I(t,E)=e^{\frac{i}{h}\,S(t,E)}\,\big(u_{t}(t,\tau;h)+h\,L_{1}u_{t}(t,\tau;h)+\mathcal{O}(h^{2})\big)
\end{equation}
with:
\begin{equation}\label{}
L_{1}u_{t}\big(t,\tau;h\big)=\sum_{n=0}^{2}\frac{2^{-(n+1)}}{i\,n!\,(n+1)!}\,\big<\left(
\begin{array}{cc}
0 & -1 \\
-1 & 0 \\
\end{array}
\right)^{-1} \left(
               \begin{array}{c}
                 D_{s} \\
                 D_{\tilde{\tau}} \\
               \end{array}
             \right)
, \left(
               \begin{array}{c}
                 D_{s} \\
                 D_{\tilde{\tau}} \\
               \end{array}
             \right)\big>^{n+1}(\phi_{t}^{n}\,u_{t})(t,\tau;h)
\end{equation}
A quick calculation shows that:
$$\big<\left(
\begin{array}{cc}
0 & -1 \\
-1 & 0 \\
\end{array}
\right)^{-1} \left(
               \begin{array}{c}
                 D_{s} \\
                 D_{\tilde{\tau}} \\
               \end{array}
             \right)
, \left(
               \begin{array}{c}
                 D_{s} \\
                 D_{\tilde{\tau}} \\
               \end{array}
             \right)\big>=2\,\frac{\partial^{2}}{\partial s\,\partial \tilde{\tau}}$$
and that:
$$\phi_{t}(s,\tilde{\tau})=\varphi_{t}(s,\tilde{\tau})-\varphi_{t}(t,\tau)-\frac{1}{2}\,\big<\left(
\begin{array}{cc}
0 & -1 \\
-1 & 0 \\
\end{array}
\right) \left(
               \begin{array}{c}
                 s-t \\
                 \tilde{\tau}-\tau\\
               \end{array}
             \right)
, \left(
               \begin{array}{c}
                 s-t \\
                 \tilde{\tau}-\tau \\
               \end{array}
             \right)\big>=0$$
It follows that:
$$L_{1}u_{t}\big(t,\tau;h\big)=-i\,(\frac{\partial^{2} u_{t}}{\partial s\,\partial \tilde{\tau}})(t,\tau;h)$$
Thus:
\begin{equation}\label{}
I(t,E)=e^{\frac{i}{h}\,S(t,E)}\,\Big[Q(\tau;h)\,\chi'(t)\,a(t,E;h)-i\,\frac{h}{2}\,\partial_{\tau} Q(\tau;h)\,\chi''(t)\,a(t,E;h)-i\,h\,\partial_{\tau} Q(\tau;h)\,\chi'(t)\,\partial_{t}a(t,E;h)+\mathcal{O}(h^{2})\Big]
\end{equation}
which gives:
\begin{align*}
(u|v)_{\chi}&=\big(I(t,E)\,u|\mathcal{K}(t,E)\,v\big)=u\,\overline{v}\,\int_{0}^{2\,\pi}I(t,E)\,\overline{\mathcal{K}(t,E)}\,dt\ \\
            &=u\,\overline{v}\,\Big(C^{2}_{0}\,f'_{0}(\tau)+h\,\big(2\,C_{0}\,C_{1}\,f'_{0}(\tau)+C^{2}_{0}\,f'_{1}(\tau)+
            C^{2}_{0}\,\widetilde{S_{1}}(E)\,f''_{0}(\tau)\big)+\mathcal{O}(h^{2})\Big)
\end{align*}
We should therefore have:
$$C^{2}_{0}\,f'_{0}(\tau)=1$$
and
$$2\,C_{0}\,C_{1}\,f'_{0}(\tau)+C^{2}_{0}\,f'_{1}(\tau)+C^{2}_{0}\,\widetilde{S_{1}}(E)\,f''_{0}(\tau)=0$$
If we choose $C_{0}>0$, we have:
\begin{equation}\label{CZERO}
C_{0}=C_{0}(\tau)=(f'_{0}(\tau))^{-1/2}
\end{equation}
and
\begin{equation}\label{CINDICE1}
C_{1}=C_{1}(\tau)=-\frac{1}{2}\,(f'_{0}(\tau))^{-1/2}\,(\frac{f_{1}}{f'_{0}})'(\tau)
\end{equation}
we end up with $(u|v)_{\chi}=u\,\overline{v}\,\big(1+\mathcal{O}(h^{2})\big)$, which normalizes $\mathcal{K}(t,E)$ to order 2.

We define $\mathcal{K}_{0}(t,E)=\mathcal{K}(t,E)$ (Poisson operator with data at $t=0$), and $\mathcal{K}_{2\,\pi}(t,E)=\mathcal{K}_{0}(t-2\,\pi,E)$ (Poisson operator with data at $t=2\,\pi$).

The energy $E$ is an eigenvalue of the operator $f(hD_{t};h)$ if and only if 1 is an eigenvalue of the Monodromy operator
$$M(E)=\mathcal{K}^{*}_{2\,\pi}(E)\,I(t,E)=\int\mathcal{K}^{*}_{0}(t-2\,\pi,E)\,I(t,E)\,dt$$
Note that in dimension 1, the monodromy operator $M(E)$ reduces to a multiplication operator.

A few lines of calculation then show that:
\begin{align}
M(E)&=\exp(2\,i\,\pi\,\tau/h)\,\exp\big(2\,i\,\pi\,\widetilde{S_{1}}(E)\big)\,\big(1+2\,i\,\pi\,h\,\widetilde{S_{2}}(E)+\mathcal{O}(h^{2})\big)\nonumber\ \\
    &=\exp\Big[\frac{i}{h}\,\big(2\,\pi\,\tau+2\,\pi\,h\,\widetilde{S_{1}}(E)+2\,\pi\,h^{2}\,\widetilde{S_{2}}(E)\big)\Big]
\end{align}
The Bohr-Sommerfeld quantization rule is written as:
$$f_{0}^{-1}(E)+h\,\widetilde{S_{1}}(E)+h^{2}\,\widetilde{S_{2}}(E)+\mathcal{O}(h^{3})=n\,h,\quad n\in \mathbb{Z}$$
Let $S_{1}(E)=2\,\pi\,\widetilde{S_{1}}(E)$, $S_{2}(E)=2\,\pi\,\widetilde{S_{2}}(E)$. Then, since $f_{0}^{-1}(E)=\tau(E)=\displaystyle \frac{1}{2\,\pi}\,\oint_{\gamma_{E}}\xi(x)\,dx=\displaystyle \frac{1}{2\,\pi}\,S_{0}(E)$, and we known that $S_{3}(E)=0$, we obtain:
\begin{equation}\label{BSACANG}
S_{0}(E)+h\,S_{1}(E)+h^{2}\,S_{2}(E)+\mathcal{O}(h^{4})=2\,\pi\,n\,h,\quad n\in \mathbb{Z}
\end{equation}
with:
\begin{equation}
S_{1}(E)=-2\,\pi\,\big(\frac{f_{1}}{f'_{0}}\big)(\tau),\quad S_{2}(E)=\frac{2\,\pi}{f'_{0}(\tau)}\,\Big(\frac{1}{2}\,\big(\frac{f_{1}^{2}}{f'_{0}}\big)'(\tau)-f_{2}(\tau)\Big)
\end{equation}
\section{The discrete spectrum of $P$ in $I$}\label{Discrete Spectrum}
Here we recover the fact that BS determines asymptotically all eigenvalues of $P$ in $I$. We adapt the argument of \cite{SjZw}. It is to think of $\{a_{E}\}$ and $\{a'_{E}\}$ as zero dimensional "Poincaré sections" of $\gamma_{E}$. Let $\mathcal{K}_{h}(E)$ be the operator (Poisson operator) that assigns to its "initial value" $C_{0}\in L^{2}\big(\{a_{E}\}\big)\approx \mathbb{R}$ the well normalized solution $u(x;h)=\displaystyle\int e^{\frac{i}{h}\,\big(x\xi+\psi(\xi)\big)}\,b(\xi;h)\,d\xi$ to $\big(P(x,hD_{x})-E\big)u=0$ near $\{a_{E}\}$. By construction, we have:
\begin{equation}\label{DSISPECt1}
\pm \big(\mathcal{L}^{a}(E)\big)^{*}\frac{i}{h}[P,\chi^{a}]_{\pm}\mathcal{L}^{a}(E)=\mathrm{Id}_{a_{E}}=1
\end{equation}
We define object "connecting" $a$ to $a'$ along $\gamma_{E}$ as follows: let $\widetilde{T}=\tilde{T}(E)>0$ such that $\exp \widetilde{T}H_{p_{0}}(a)=a'$ \big(in case $p_{0}$ is invariant by time reversal, i.e. $p_{0}(x,\xi)=p_{0}(x,-\xi)$ we take $\widetilde{T}(E)=T(E)/2$\big).

Choose $\chi_{f}^{a}$ ($f$ for "forward") be a cut-off function supported microlocally near $\gamma_{E}$, equal to 0 along $\exp tH_{p_{0}}(a)$ for $t\leq \varepsilon$, equal to 1 along $\gamma_{E}$ for $t\in[2\varepsilon, \widetilde{T}+\varepsilon]$, and back to 0 next to $a'$, e.g. for $t\geq \widetilde{T}+2\varepsilon$. Let similarly $\chi_{b}^{a}$ ($b$ for "backward") be a cut-off function supported microlocally near $\gamma_{E}$, equal to 1 along $\exp tH_{p_{0}}(a)$ for $t\in[-\varepsilon, \widetilde{T}-2\varepsilon]$, and equal to 0 next to $a'$, e.g. for $t\geq\widetilde{T}-\varepsilon$. By (\ref{DSISPECt1}) we have
\begin{equation}\label{DSISPECt2}
\big(\mathcal{L}^{a}(E)\big)^{*}\frac{i}{h}[P,\chi^{a}]_{+} \mathcal{L}^{a}(E)=
\big(\mathcal{L}^{a}(E)\big)^{*}\frac{i}{h}[P,\chi^{a}_{f}] \mathcal{L}^{a}(E)=1
\end{equation}
\begin{equation}\label{DSISPECt3}
-\big(\mathcal{L}^{a}(E)\big)^{*}\frac{i}{h}[P,\chi^{a}]_{-}\mathcal{L}^{a}(E)=-
\big(\mathcal{L}^{a}(E)\big)^{*}\frac{i}{h}[P,\chi^{a}_{b}] \mathcal{L}^{a}(E)=1
\end{equation}
which define a left inverse $\mathcal{R}^{a}_{+}(E)=\displaystyle\big(\mathcal{L}^{a}(E)\big)^{*}\frac{i}{h}[P,\chi^{a}_{f}]$ to $\mathcal{L}^{a}(E)$ and a right inverse $\mathcal{R}^{a}_{-}(E)=-\displaystyle\frac{i}{h}[P,\chi^{a}_{b}] \mathcal{L}^{a}(E)$ to $\big(\mathcal{L}^{a}(E)\big)^{*}$. We define similar objects connecting $a'$ to $a$, $\widetilde{T}'=\widetilde{T}'(E)>0$ such that $\exp \widetilde{T}'H_{p_{0}}(a)=a'$ \big($\widetilde{T}=\widetilde{T}'$ if $p_{0}$ is invariant by time reversal\big), in particular a left inverse $\mathcal{R}^{a'}_{+}(E)=\displaystyle\big(\mathcal{L}^{a'}(E)\big)^{*}\frac{i}{h}[P,\chi^{a'}_{f}]$ to $\mathcal{L}^{a'}(E)$ and a right inverse $\mathcal{R}^{a'}_{-}(E)=-\displaystyle\frac{i}{h}[P,\chi^{a'}_{b}] \mathcal{L}^{a'}(E)$ to $\big(\mathcal{L}^{a'}(E)\big)^{*}$, with the additional requirement
\begin{equation}\label{DSISPECt4}
\chi^{a}_{b}+\chi^{a'}_{b}=1
\end{equation}
near $\gamma_{E}$. Define now the pair $\mathcal{R}_{+}(E)u=\big(\mathcal{R}^{a}_{+}(E)u,\,\mathcal{R}^{a'}_{+}(E)u\big)$, $u\in L^{2}(\mathbb{R})$ and $\mathcal{R}_{-}(E)$ by $\mathcal{R}_{-}(E)u_{-}=\mathcal{R}^{a}_{-}(E)u^{a}_{-}+\mathcal{R}^{a'}_{-}(E)u^{a'}_{-}$,\,$u_{-}=(u^{a}_{-},\,u^{a'}_{-})\in \mathbb{C}^{2}$, we call \emph{Grushin operator} $\mathcal{P}(z)$ the operator defined by the linear system
\begin{equation}\label{DSISPECt5}
\frac{i}{h}(\mathcal{P}-z)u+\mathcal{R}_{-}(z)u_{-}=v,\quad \mathcal{R}_{+}(z)u=v_{+}
\end{equation}
From \cite{SjZw}, we know that the problem (\ref{DSISPECt5}) is well posed, and
\begin{equation}\label{DSISPECt6}
\big(\mathcal{P}(z)\big)^{-1}=\left(
                                \begin{array}{cc}
                                  \mathcal{E}(z) & \mathcal{E}_{+}(z) \\
                                  \mathcal{E}_{-}(z) & \mathcal{E}_{-+}(z) \\
                                \end{array}
                              \right)
\end{equation}
with $(P-z)^{-1}=\mathcal{E}(z)-\mathcal{E}_{+}(z)\,\big(\mathcal{E}_{-+}(z)\big)^{-1}\,\mathcal{E}_{-}(z)$. Actually one can show that the effective Hamiltonian $\mathcal{E}_{-+}(z)$ is singular precisely when 1 belongs to the spectrum of the monodromy operator, or when the
microlocal solutions $u_{1}, u_{2}\in \mathcal{K}_{h}(E)$ computed in (\ref{2.14}) are colinear, which amounts to say that Gram matrix (\ref{MaTrIcEgRaMm}) is singular. There follows that the spectrum of $P$ in $I$ is precisely the set of $z$ we have determined by BS quantization rule.\newpage
\section*{Appendix. Weyl quantization and conjugation by Fourier integral operators}
To fix the ideas we choose the real case (analytic or $C^{\infty}$). To start with, we consider the simple case of conjugation by an elleptic factor. Let $p(x,\xi;h)\sim p_{0}(x,\xi)+h\,p_{1}(x,\xi)+\cdots$ be a classical symbol of order 0, defined near $(x_{0},\xi_{0})\in \mathbb{R}^{2}$. Let $\varphi(x)$ be a real valued smooth function, defined near $x_{0}$, such that $\varphi'(x_0)=\zeta_0$. Let $P(x,hDx;h)=p^{w}(x,hDx;h)$ be the Weyl $h$-quantization of $p(x,\xi;h)$. We are then interested in the Weyl symbol of the pseudodifferential operator $Q=e^{-\frac{i}{h}\,\varphi(x)}\circ P \circ e^{\frac{i}{h}\,\varphi(x)}$, which is defined near $(x_0,\xi_0-\zeta_0)$. We proceed formally, by first writing the integral kernel of $Q$ as,
\begin{equation}\label{Noyau1}
K_{Q}(x,y)=\int e^{\frac{i}{h}\,
\big((x-y)\,\theta+\varphi(y)-\varphi(x)\big)}\,p\big(\frac{x+y}{2},\theta;h\big)\,
\widetilde{d\theta},\quad\widetilde{d\theta}=(2\,\pi\,h)^{-1}\,d\theta
\end{equation}
It is well known (Kuranishi's Trick) that
\begin{equation*}
\varphi(x)-\varphi(y)=F(x,y)\,(x-y)
\end{equation*}
with
\begin{equation*}
F(x,y)=\varphi'_{x}(\frac{x+y}{2})+\mathcal{O}(x-y)
\end{equation*}
The change of variables $\theta\mapsto \theta-F(x,y)$ then allows us to write
\begin{equation}\label{Noyau2}
K_Q(x,y)=\int e^{\frac{i}{h}\,(x-y)\,\theta}\,p\big(\frac{x+y}{2},\theta+F(x,y);h\big)
\,\widetilde{d\theta}
\end{equation}
We have
\begin{equation*}
p\big(\frac{x+y}{2},\theta+F(x,y);h\big)=
p\big(\frac{x+y}{2},\theta+\varphi'_{x}(\frac{x+y}{2});h\big)+r\big(x,y,\theta\big)\,(x-y)^{2}
\end{equation*}
where
\begin{equation*}
r(x,y,.)\in \mathcal{S}(\mathbb{R})
\end{equation*}
By integrating by parts, we obtain:
\begin{equation*}
\int e^{\frac{i}{h}\,(x-y)\,\theta}\,r(x,y,\theta)\,(x-y)^{2}
\,\widetilde{d\theta}=\mathcal{O}(h^{2})
\end{equation*}
Here the $\mathcal{O}$ term in (\ref{Noyau2}) contributes to the Weyl symbol of $Q$ by $\mathcal{O}(h^{2})$ (i.e., by a classical symbol of order -2). Therefore, if we denote by $q$ the Weyl symbol of $Q$, we have:
\begin{equation*}
q(x,\xi;h)=p\big(x,\xi+\varphi'(x);h\big)+\mathcal{O}(h^{2})
\end{equation*}
If $a=a(x;h)\sim a_{0}(x)+h\,a_{1}(x)+\cdots$ is a classical symbol defined near $x_{0}$, then:
\begin{equation}
\begin{aligned}\label{QAa01}
Qa(x;h)&=\big(e^{-\frac{i}{h}\,\varphi(x)}\circ P \circ e^{\frac{i}{h}\,\varphi(x)}\big)\big(a(x;h)\big)=e^{-\frac{i}{h}\,\varphi(x)}P\big(e^{\frac{i}{h}\,\varphi(x)}a(x;h)\big)\ \\
       &=(2\pi h)^{-1}\,\int\int e^{\frac{i}{h}\,(x-y)\,\theta}\,p\big(\frac{x+y}{2},\theta+F(x,y);h\big)\,a(y;h)\,dy\,d\theta
\end{aligned}
\end{equation}
Thus, the operator $e^{-\frac{i}{h}\,\varphi(x)}\circ P \circ e^{\frac{i}{h}\,\varphi(x)}$ is a semi-classical pseudo-differential operator.

Recall that in dimension 1, we have:
\begin{equation}\label{A1}
(2\pi h)^{-1}\,\int\int e^{-\frac{i}{h}\,z\,\theta}\,u(z,\theta)\,dz\,d\theta\sim\sum_{k=0}^{2}\frac{h^{k}}{k!\,i^{k}}\big((\partial_{z} \partial_{\theta})^{k}u\big)(0,0)+\mathcal{O}(h^{3})
\end{equation}
In (\ref{QAa01}), we perform the change of variables $(z, \theta)=(y-x, \theta)$. The Jacobian of this transformation is equal to 1.

Now, for a fixed $x$, we expand using the stationary phase formula (\ref{A1}), yielding:
\begin{equation}\label{A3}
Qa(x;h)=(2\pi h)^{-1}\,\int\int e^{-\frac{i}{h}\,z\,\theta}\,p\big(x+\frac{z}{2},\theta+F(x,z+x);h\big)\,a(z+x;h)\,dz\,d\theta
\end{equation}
Let
$$u(z,\theta)=p\big(x+\frac{z}{2},\theta+F(x,z+x);h\big)\,a(z+x;h)$$
We then have:
$$Qa(x;h)\sim u(0,0)+\frac{h}{i}\,(\frac{\partial^{2} u}{\partial z\,\partial \theta})(0,0)-\frac{h^{2}}{2}\,
(\frac{\partial^{4} u}{\partial z^{2}\,\partial \theta^{2}})(0,0)+\mathcal{O}(h^{3})$$
Moreover:
\begin{equation*}
F(x,x)=\int_{0}^{1}\varphi'(x)\,dt=\varphi'(x)
\end{equation*}
\begin{equation*}
\partial_{z}\big(F(x, z+x)\big)_{z=0}=\int_{0}^{1}(1-t)\,\varphi''(x)\,dt=\frac{1}{2}\,\varphi''(x)
\end{equation*}
\begin{equation*}
\frac{\partial^{2}}{\partial z^{2}}\big(F(x, z+x)\big)_{z=0}=\int_{0}^{1}(1-t)^{2}\,\varphi'''(x)\,dt=\frac{1}{3}\,\varphi'''(x)
\end{equation*}
A straightforward calculation shows that:
$$u(0,0)=p\big(x,\varphi'(x);h\big)\,a(x;h)$$
$$(\frac{\partial^{2} u}{\partial z\,\partial \theta})(0,0)=\beta(x;h)\,\partial_{x}a(x;h)+\frac{1}{2}\,\partial_{x}\beta(x;h)\,a(x;h)$$
where we have defined
$$\beta(x;h)=(\frac{\partial p}{\partial \xi})\big(x,\varphi'(x);h\big)$$
and
$$(\frac{\partial^{4} u}{\partial z^{2}\,\partial \theta^{2}})(0,0)=\frac{1}{4}\,\partial_{x}r(x;h)\,a(x;h)+\frac{1}{4}\,\varphi''(x)\,\partial_{x}\theta(x;h)\,a(x;h)+
\partial_{x}\gamma(x;h)\,\partial_{x}a(x;h)+\gamma(x;h)\,\frac{\partial^{2} a(x;h)}{\partial x^{2}}+\frac{1}{3}\,\varphi'''(x)\,
\theta(x;h)\,a(x;h)$$
where we have set
\begin{equation*}
r(x;h)=(\frac{\partial^{3} p}{\partial x \partial \xi^{2}})\big(x,\varphi'(x)\big),\quad\gamma(x;h)=(\frac{\partial^{2} p}{\partial \xi^{2}})\big(x,\varphi'(x);h\big),\quad\theta(x;h)=(\frac{\partial^{3} p}{\partial \xi^{3}})\big(x,\varphi'(x);h\big)
\end{equation*}
Thus, we obtain modulo $\mathcal{O}(h^{3})$:
\begin{equation}\label{05OSCI0}
\begin{aligned}
&(Q-E)\big(a(x;h)\big)=
\Big(p\big(x,\varphi'(x);h\big)-E\Big)\,a(x;h)
+\frac{h}{i}\,\Big(\beta(x;h)\,\partial_{x}a(x;h)+\frac{1}{2}\,\partial_{x}\beta(x;h)\,a(x;h)\Big)\cr
&-h^{2}\,\Big(\frac{1}{8}\,\partial_{x}r(x;h)\,a(x;h)+\frac{1}{8}\,\varphi''(x)\,\partial_{x}\theta(x;h)\,a(x;h)+\frac{1}{2}\,
\partial_{x}\gamma(x;h)\,\partial_{x}a(x;h)+\frac{1}{2}\,\gamma(x;h)\,\frac{\partial^{2} a(x;h)}{\partial x^{2}}+\frac{1}{6}\,\varphi'''(x)\,
\theta(x;h)\,a(x;h)\Big)
\end{aligned}
\end{equation}
\nocite{*}
\bibliographystyle{plain}
\bibliography{BSQR2025}
\addcontentsline{toc}{chapter}{Bibliographie}
\end{document}